\pdfoutput=1
\documentclass[journal]{IEEEtran}
\usepackage[utf8]{inputenc}
\usepackage{setspace}
\usepackage{amsmath}
\usepackage{amsfonts}
\usepackage{bbding}
\usepackage{amssymb}
\usepackage{array}
\usepackage{multirow}
\usepackage{graphicx}
\usepackage{algorithm}
\usepackage{algorithmic}
\newcommand\algorithmicprocedure{\textbf{procedure}}
\newcommand{\algorithmicendprocedure}{\algorithmicend\ \algorithmicprocedure}
\makeatletter
\newcommand\PROCEDURE[3][default]{%
  \ALC@it
  \algorithmicprocedure\ \textsc{#2}(#3)%
  \ALC@com{#1}%
  \begin{ALC@prc}%
}
\newcommand\ENDPROCEDURE{%
  \end{ALC@prc}%
  \ifthenelse{\boolean{ALC@noend}}{}{%
    \ALC@it\algorithmicendprocedure
  }%
}
\newenvironment{ALC@prc}{\begin{ALC@g}}{\end{ALC@g}}
\makeatother

\usepackage{color}
\usepackage{cite}

\allowdisplaybreaks

\DeclareMathOperator{\cov}{\bf cov}

\DeclareMathOperator{\expt}{\mathbb{E}}
\DeclareMathOperator{\vari}{\mathbb{V}}
\DeclareMathOperator{\proj}{\textrm{proj}}

\usepackage{amsthm}
\usepackage{stfloats}
\usepackage{caption}
\usepackage{subcaption}

\newtheorem{theorem}{Theorem}
\newtheorem{lemma}{Lemma}

\newtheorem{cor}{Corollary}

\begin{document}
\title{Opportunistic Scheduling Using Statistical Information of Wireless Channels}

\author{
  \IEEEauthorblockN{Zhouyou Gu,  Wibowo Hardjawana, Branka Vucetic}\\
    \thanks{
    The work of Zhouyou Gu was in part by an Australian Government Research Training Program Scholarship, in part by the University of Sydney's Supplementary and Completion Scholarships and in part by the Australian Research Council (ARC) Discovery grant number DP210103410.
    The work of Wibowo Hardjawana was supported in part by the ARC Discovery grant number DP210103410.
    The work of Branka Vucetic was supported in part by the ARC Laureate Fellowship grant number FL160100032 and in part by the ARC Discovery grant number DP210103410.
    \emph{(Corresponding author: Zhouyou Gu.)}}
    \thanks{Z. Gu, W. Hardjawana, and B. Vucetic are with the School of Electrical and Information Engineering, University of Sydney, Sydney, NSW 2006, Australia (email: \{zhouyou.gu, wibowo.hardjawana, branka.vucetic\}@sydney.edu.au).}
    }
\maketitle
\begin{abstract}
This paper considers opportunistic scheduler (OS) design using statistical channel state information~(CSI). We apply max-weight schedulers (MWSs) to maximize a utility function of users' average data rates. MWSs schedule the user with the highest weighted instantaneous data rate every time slot.
Existing methods require hundreds of time slots to adjust the MWS's weights according to the instantaneous CSI before finding the optimal weights that maximize the utility function.
In contrast, our MWS design requires few slots for estimating the statistical CSI. Specifically, we formulate a weight optimization problem using the mean and variance of users' signal-to-noise ratios (SNRs) to construct constraints bounding users' feasible average rates. Here, the utility function is the formulated objective, and the MWS's weights are optimization variables. We develop an iterative solver for the problem and prove that it finds the optimal weights.
We also design an online architecture where the solver adaptively generates optimal weights for networks with varying mean and variance of the SNRs.
Simulations show that our methods effectively require $4\sim10$ times fewer slots to find the optimal weights and achieve $5\sim15\%$ better average rates than the existing methods.
\end{abstract}
\section{Introduction}
Time-varying channels limit the performance of multi-user wireless networks \cite{viswanath2002opportunistic}. 
Scheduling strategies for users' transmissions according to the stochastic variation of channel states is the key to optimizing the long-term system objectives of wireless networks.
These objectives are framed as utility functions specifically designed to track networks' performance metrics \cite{asadi2013survey}, such as fairness or maximization of data rates.
Such scheduling strategies are referred to as opportunistic schedulers~(OSs). 
OSs schedule users with good channels and do not schedule those users when their channel qualities are bad. Thus, OSs can achieve higher spectrum efficiency than other schedulers considering no channel variation, e.g., random and round-robin schedulers \cite{ouyang2015exploiting,kawser2012performance}.
Two OS classes have been reported in the open literature \cite{asadi2013survey,neely2019convergence}, Markov decision process (MDP)-based OSs \cite{zhang2008approximate,ouyang2015exploiting,borkar2017opportunistic,kadota2018scheduling,chen2021bringing,wang2018deep,gu2021knowledge} and max-weight schedulers (MWSs) \cite{tse1999transmitter,agrawal2002optimality,borst2001dynamic,liu2003framework,neely2008fairness,mandelli2019satisfying}.

MDP-based OSs maximize the long-term utility function by calculating the optimal selection of users to be scheduled in each channel state, assuming the full prior knowledge of the statistical channel state information (CSI), e.g., the transition probabilities of the channel states \cite{zhang2008approximate,ouyang2015exploiting,borkar2017opportunistic,kadota2018scheduling}.
The calculated optimal user selections are saved in a lookup table that is referred to for the channel state in each time slot \cite{sutton2011reinforcement}.
Alternative to the above tabular approach, the MDP-based OS also uses a neural network (NN) to map the channel state into the optimal user selection in each slot. The NN's parameters can then be stochastically optimized by deep reinforcement learning (DRL) based on the channel states, user selections and well-designed reward signals in every slot \cite{chen2021bringing,wang2018deep,gu2021knowledge} without any knowledge of the statistical CSI. As the NN contains many parameters that require optimization, DRL methods take a long time to find the optimal MDP-based OS \cite{gu2021knowledge}.

The second OS class, namely MWSs \cite{tse1999transmitter,agrawal2002optimality,borst2001dynamic,liu2003framework,neely2008fairness,mandelli2019satisfying}, schedules a user with the highest weighted instantaneous utility, e.g., the highest weighted instantaneous rate, in each time slot. 
These methods continuously adjust the MWS weights in each slot based on every past channel state and user selection to maximize the long-term utility function. This approach requires no prior knowledge of the statistical CSI.
Since MWSs only need to optimize a vector of weights, they have a much lower implementation complexity than the lookup table or the NN of MDP-based OSs mentioned above.
Unfortunately, these MWS approaches still require hundreds of time slots in trials of adjusting weights before they find optimal ones, leading to suboptimal system performance during the weight adjustment.
Applying the statistical CSI in the MWS design can save time slots in the online weights adjustment \cite{neely2019convergence,huang2014power}. 
For example, an online algorithm is developed in \cite{huang2014power} to measure the full probability distribution of discrete channel states. This algorithm calculates the optimal MWS weights to maximize the expected utility function based on the measured distribution. Therefore, it quickly converges to the optimal weights, particularly when the number of possible channel states is limited. 
However, in practice, the channel states are measured as continuous numbers \cite{boussad2021evaluating}. Estimating their distributions requires discretization and numerous time slots to ensure all states are sufficiently counted.
It requires further investigation of the MWS design that leverages the prior statistical information of continuous channel states to reduce the time cost.

In this paper, we propose a new method that uses the limited prior knowledge of the statistical CSI to effectively reduce the number of time slots required in the MWS design.
We find the optimal MWS weights to maximize the utility function as the sum of the logs of users' average scheduled bit rates in the multi-user wireless network \cite{asadi2013survey,neely2019convergence,tse1999transmitter,agrawal2002optimality}.
In this work, users' signal-to-noise ratios (SNRs) are considered as the CSI, which can be measured as continuous variables from radio signals, e.g., 5G networks' CSI reference signals \cite{boussad2021evaluating,3gpp.38.214}. 
We first derive each user's average rate for given MWS weights from the full prior knowledge
of the statistical CSI, namely the probability density functions
(PDFs) of users' SNRs \cite{liu2011asymptotic,holtzman2001asymptotic,choi2015throughput}.
However, the full distribution of SNRs requires a significant amount of time to measure, and the computation of users' rates is difficult because it needs to compute the integral of the distributions.
To use the limited prior knowledge of the statistical CSI instead, we re-derive the computation of users' average rates for given MWS weights as an optimization problem only based on the mean and variance (i.e., the first and second moment of the PDF \cite{feller2008introduction}) of users' SNRs, referred to as a rate estimation problem.
Here, we use the mean and variance of SNRs to construct constraints that bound users' feasible average data rates.
Next, we formulate mean-variance-based weight optimization (MVWO) that maximizes the above utility function.
We construct this problem as a bi-level optimization problem (BLOP) \cite{colson2007overview,sinha2017review} with the MWS weights as optimization variables. The rate estimation problem is embedded in the BLOP to specify the average rates at given weights.
We design an iterative solver for the BLOP and mathematically prove that it returns the optimal MWS weights in MVWO.
Furthermore, since real-world networks have time-varying mean and variance of SNRs, e.g., due to the mobility of users, we study how to use the proposed MVWO method to adjust the MWS weights based on online SNR measurements.

We summarize the main contributions of this paper to the literature as follows,
\begin{itemize}
\item
To the best of our knowledge, this work is the first that proposes to design MWSs based on the limited prior knowledge of the statistical CSI, namely the mean and variance of users' SNRs, which costs a few samples to estimate. Our methods can be applied to continuous channel states by estimating their mean and variance. This is unlike the existing channel-statistics-based method \cite{huang2014power} that estimates the whole distribution of the channel states, requiring a much larger number of samples (or time slots) for the estimation.
Also, the proposed method reduces the time complexity (i.e., the number of time slots required) in optimizing the MWS weights $4\sim10$ times compared to the stochastic methods \cite{tse1999transmitter,agrawal2002optimality} that use no prior knowledge of the statistical CSI, as shown by simulations.

\item
We design a new iterative solver to solve the BLOP in MVWO, which has less computational complexity than existing iterative BLOP solvers \cite{kolstad1990derivative,savard1994steepest,colson2007overview,sinha2017review}.
Specifically, the designed solver updates weights in each iteration via normalization and linear combinations of vectors rather than solving an optimization problem in each iteration as the existing iterative BLOP solvers do. It reduces the complexity of weight updates from a polynomial computational complexity in the existing solvers to a linear one in our solver.

\item
We mathematically prove that the solution of the BLOP solver converges to the optimal weights with a convergence error $\hat{\epsilon}$ in $O(\frac{1}{\hat{\epsilon}^2}K\log K)$ iterations, where $K$ is the number of users and $\hat{\epsilon}$ describes how close the solution satisfies the optimality condition of weights.
Note that the optimal weights maximize the utility function in the BLOP.
The simulation results show that the designed iterative solver converges
to the optimal weights in tens of iterations at a high probability, e.g., $90\sim100\%$, when there are up to ten users.

\item
We design an online architecture where the proposed solver continuously adjusts the MWS weights based on the mean and variance of the SNRs measured online.
The simulation results show that the designed architecture achieves $5\sim15\%$ better performance than the existing MWS approaches \cite{tse1999transmitter,agrawal2002optimality} in terms of the geometrical mean of users' average rates (an equivalent expression of the studied utility function).
\end{itemize}

\subsection{Notations}
$\langle\mathbf{x},\mathbf{y}\rangle$ denotes the inner product of $\mathbf{x}$ and $\mathbf{y}$. $\|\mathbf{x}\|_2$ denotes the $\ell_2$-norm of $\mathbf{x}$.
$\mathbf{x}\odot\mathbf{y}$ and $\mathbf{x}\oslash\mathbf{y}$ are element-wise multiplication and division between $\mathbf{x}$ and $\mathbf{y}$, respectively. 
For a $K$-dimensional non-negative/positive vector $\mathbf{x}$, we write $\mathbf{x}>0$ when $x_k>0$ $\forall k =1,\dots,K$, and $\mathbf{x}\geq0$ when $x_k\geq 0$ $\forall k =1,\dots,K$.
We write a tuple as $(\mathbf{x}^{(1)},\dots,\mathbf{x}^{(i)},\dots)$ where $\mathbf{x}^{(i)}$ is the $i$-th element in the tuple. We obtain the $\mathbf{x}^{(i)}$ from the tuple as $\mathbf{x}^{(i)}=\proj_i[(\mathbf{x}^{(1)},\dots,\mathbf{x}^{(i)},\dots)]$.
The geometric mean of elements in a $K$-dimensional vector, $\mathbf{x}$, is expressed as $\textrm{GM}(\mathbf{x}) = (\prod_{k=1}^{K}x_k)^{\frac{1}{K}}$, where $x_k$ is the $k$-th element of $\mathbf{x}$.
A $K\times K$ positive semidefinite matrix $\mathbf{X}$ is denoted as $\mathbf{X} \succeq 0$, where positive semidefiniteness implies $\mathbf{z}^{\rm T}\mathbf{X}\mathbf{z} \geq 0$, $\forall \mathbf{z} \in \mathbb{R}^{K}$. We also summarize the frequently used symbols in Table \ref{tab:symbols}.

\begin{table}[t]
\caption{Summary of Symbols}\label{tab:symbols}
\small
\begin{tabular}{c|l}
\hline\hline
$K$ & number of users \\
$\Delta_\textrm{0}$ & duration of a time slot \\
$B$ & bandwidth of the channel \\
$\mathbf{x}(t)$ & user scheduling actions in the $t$-th slot \\
$x_k(t)$ & user $k$'s scheduling action in the $t$-th slot\\
$\mathbf{s}(t)$ & users' channel states (or SNRs) in the $t$-th slot \\
$\phi_k(t)$ & user $k$'s SNR in the $t$-th slot\\
$\mu(\cdot|\mathbf{w})$ & MWS with weights $\mathbf{w}$ \\
$\mathbf{r}^{\sim\mu(\cdot|\mathbf{w})}$ & users' data rates achieved by the MWS $\mu(\cdot|\mathbf{w})$\\
$r_k^{\sim\mu(\cdot|\mathbf{w})}$ & user $k$'s data rate achieved by the MWS $\mu(\cdot|\mathbf{w})$\\
$f(\cdot)$ & utility function of users' rates\\
$\mathcal{F}$ & feasible rate region of users\\
$m^{\phi}_{k}$ & mean of user $k$'s SNR\\
$v^{\phi}_{k}$ & variance of user $k$'s SNR\\
$\mathbf{r}$ & users' rates achieved by any valid scheduler\\
$\mathbf{p}$ & expected value of $\mathbf{x}(t)$ over time\\
$\mathbf{y}$ & expected value of users' SNRs in scheduled slots\\
$\mathbf{H}$ & covariance between $\mathbf{x}(t)$ and $\mathbf{s}(t)$\\
$\mathcal{G}$ & estimated feasible rate region of users\\
$\mathbf{r}_\mathcal{G}^{\sim\mu(\cdot|\mathbf{w})}$ & estimated rates $\mathbf{r}^{\sim\mu(\cdot|\mathbf{w})}$ achieved by $\mu(\cdot|\mathbf{w})$\\
$\mathbf{r}^*$ & optimal $\mathbf{r}$ in $\mathcal{G}$ that maximizes $f(\mathbf{r})$\\
$\mathbf{w}^*$ & optimal $\mathbf{w}$ that maximizes $f(\mathbf{r}_\mathcal{G}^{\sim\mu(\cdot|\mathbf{w})})$ \\
$\mathbf{w}^{(i)}$ & MWS's weights in the $i$-th iteration of the solver\\
$\mathbf{r}^{(i)}$ & estimated rates $\mathbf{r}^{\sim\mu(\cdot|\mathbf{w}^{(i)})}$ achieved by $\mu(\cdot|\mathbf{w}^{(i)})$ \\
$\mathbf{u}^{(i)}$ & intermediary vector to compute $\mathbf{w}^{(i+1)}$ \\
$a^{(i)}$, $b^{(i)}$& scaling factors of $\mathbf{w}^{(i)}$ when constructing $\mathbf{w}^{(i+1)}$  \\
$\hat{\epsilon}$ & convergence error of the solver\\
$\hat{R}$ & maximum distance between any two vectors in $\mathcal{G}$ \\
\hline\hline

\end{tabular}
\end{table}

\section{System Model and Problem Formulation of Opportunistic Scheduling}\label{sec:system_model}
\begin{figure}[t]
\begin{subfigure}[b]{1\columnwidth}
\centering
\includegraphics[scale=0.9]{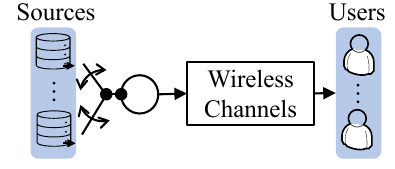}
\caption{}
\label{fig:op_sch}
\end{subfigure}
\begin{subfigure}[b]{1\columnwidth}
\centering
\includegraphics[scale=0.9]{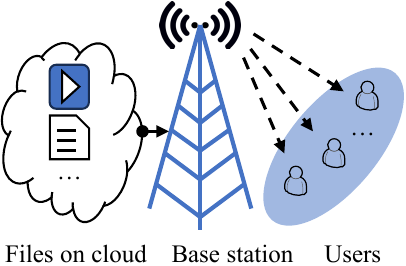}
\caption{}
\label{fig:application_illustration}
\end{subfigure}
\caption{(a) Illustration of a wireless scheduler, (b) illustration of the application scenarios.}
\vspace{-0.2cm}
\end{figure}

This section presents the system model of the multi-user wireless network and formulates the MWS design problem. It also studies the issues of using PDFs of users' SNRs to design MWSs.
\subsection{System Model}
We consider a wireless link of the base station (BS) shared by $K$ users in time slots to download data from data sources, as shown in Fig. \ref{fig:op_sch}. 
We assume that each user is using the enhanced mobile broadband service \cite{popovski20185g} to download a large file or video, which requires the BS to transmit a significant amount of data over time, as shown in Fig. \ref{fig:application_illustration}. In such scenarios, we can assume that users have infinite backlogs at the base station \cite{chen2013measurement}. This means the BS always has data ready to be transmitted for each user. Also, we do not consider any queuing states (e.g., queuing delay).
The duration of each slot in the system is $\Delta_\text{0}$ in seconds. The bandwidth of the link is $B$ in Hertz. A binary indicator of the user selection decision, $x_k(t)$, represents whether the BS transmits user $k$'s data, $k=1,\dots,K$, in the $t$-th slot or not, $t=1,2,\dots$, i.e.,
\begin{equation}\label{eq:const:channel_access_binary}
\begin{aligned}
x_k(t) \in \{0,1\}\ ,\ \forall t, k \ .
\end{aligned}
\end{equation}
For example, $x_k(t)=1$, if user $k$ occupies the $t$-th slot and the BS transmits this user's data, otherwise $x_k(t)=0$.
We assume that only one user can access the channel in each slot as
\begin{equation}\label{eq:const:channel_access}
\begin{aligned}
\sum_{k=1}^{K} x_k(t) \leq 1 \ ,\ \forall t \ .
\end{aligned}
\end{equation}
The user scheduling actions for all users in the $t$-th slot are defined as
\begin{equation}\label{eq:defi:x}
\begin{aligned}
\mathbf{x}(t)\triangleq[x_1(t),\dots,x_K(t)]^{\rm T} \ .
\end{aligned}
\end{equation}
We assume that all users' channels have stationary SNRs that are continuous random variables. User $k$'s SNR in the $t$-th slot is denoted as $\phi_k(t)$ and is assumed to be i.i.d. in each slot, $\forall k$.
The channel state of each slot is defined as
\begin{equation}\label{eq:defi:s}
\begin{aligned}
\mathbf{s}(t)\triangleq[\phi_1(t),\dots,\phi_K(t)]^{\rm T} \ , \forall t \ .
\end{aligned}
\end{equation}
SNRs of different users are assumed to be independent.
The spectrum efficiency of user $k$ in the $t$-th slot is calculated by using Shannon capacity as $\log_2\big(1+\phi_k(t)\big)$.
The amount of instantaneous bits scheduled for user $k$ in the $t$-th slot is then given as $x_k(t)\Delta_\text{0} B \log_2\big(1+\phi_k(t)\big)$.

\subsection{MWS Design Problem Formulation}
We use MWSs as OSs to schedule users and decide $\mathbf{x}(t)$ at each time slot. 
We first define $\mathbf{w} \triangleq [w_1,\dots,w_K]^{\rm T}$ as weights of the MWS for $K$ users, which are fixed over $T$ slots, and we assume the weights are normalized and positive as
\begin{equation}\label{eq:weight_non_negative_normalized}
\begin{aligned}
\|\mathbf{w}\|_2=1 \ , \ \mathbf{w} > 0 .
\end{aligned}
\end{equation}
The MWS $\mu(\cdot|\mathbf{w})$ with weights $\mathbf{w}$ decides the user scheduling action $\mathbf{x}(t)$ based on the channel state $\mathbf{s}(t)$ at the $t$-th slot as 
\begin{equation}\label{eq:max_weight_scheduler}
\begin{aligned}
\forall t,\ &\mathbf{x}(t) = [x_1(t),\dots,x_K(t)]^{\rm T} = \mu\big(\mathbf{s}(t)|\mathbf{w}\big) \\
\triangleq & \arg\max_{\mathbf{x}'(t)}\sum_{k=1}^{K} x'_k(t)  w_k \Delta_\text{0} B \log_2\big(1+\phi_k(t)\big) \ \text{s.t.}\  \eqref{eq:const:channel_access_binary}\ \eqref{eq:const:channel_access} \ .
\end{aligned}
\end{equation}
In each slot, the MWS in \eqref{eq:max_weight_scheduler} schedules user $k$ with the highest weighted instantaneous bit rate $w_k \Delta_\text{0} B \log_2\big(1+\phi_k(t)\big)$ than other users, and it sets the scheduling actions as $x_k(t)=1$ and $x_{k'}(t)=0$, $\forall k'\neq k$. Here, $\mathbf{w}$ defines the MWS's criterion of selecting the user in each slot.
The average rates achieved by MWSs for users $1,\dots,K$, $\mathbf{r}^{\sim\mu(\cdot|\mathbf{w})}\triangleq[r_1^{\sim\mu(\cdot|\mathbf{w})},\dots,r_K^{\sim\mu(\cdot|\mathbf{w})}]^{\rm T}$, can then be written as
\begin{equation}\label{eq:defi:average_mws_rate}
\begin{aligned}
&r_k^{\sim\mu(\cdot|\mathbf{w})}
= \lim_{T\to \infty} \frac{1}{T} \sum_{t=1}^{T} x_k(t) \Delta_\text{0} B \log_2\big(1+\phi_k(t)\big) \\
= &\expt[x_k(t) \Delta_\text{0} B \log_2\big(1+\phi_k(t)\big) | \mathbf{x}(t) =  \mu\big(\mathbf{s}(t)\big|\mathbf{w}\big),\forall t] , 
\end{aligned}
\end{equation}
where $\mu\big(\mathbf{s}(t)\big|\mathbf{w}\big)$ is defined in \eqref{eq:max_weight_scheduler}. The utility function of the system is the sum of the logs (or proportional fairness \cite{tse2005fundamentals,kelly1998rate}) of users' average rates expressed as
\footnote{Note that this work considers proportional fairness of users' data rates as the objective, while our methods designed later can also be applied to other general concave objective functions of users' rate (assuming maximization of the objective function), e.g., max-min fairness and $\alpha$-fairness \cite{mo2000fair}.}
\begin{equation}\label{eq:defi:utility_function_f}
\begin{aligned}
f(\mathbf{r}^{\sim\mu(\cdot|\mathbf{w})})\triangleq \sum_{k=1}^{K}\ln r_k^{\sim\mu(\cdot|\mathbf{w})}\ ,
\end{aligned}
\end{equation}
and it is strictly increasing and concave.
The problem of finding the optimal MWS weights that maximize \eqref{eq:defi:utility_function_f} is
\begin{equation}\label{eq:prob:stochastic_weight_optimization_mu_w}
\begin{aligned}
\max_{\mathbf{w}} f(\mathbf{r}^{\sim\mu(\cdot|\mathbf{w})})\ , \ \textrm{s.t.} \ \eqref{eq:weight_non_negative_normalized} ,\ \eqref{eq:defi:average_mws_rate} \ .
\end{aligned}
\tag{P1}
\end{equation}
To solve \eqref{eq:prob:stochastic_weight_optimization_mu_w}, it needs the numerical value of $\mathbf{r}^{\sim\mu(\cdot|\mathbf{w})}$ at given $\mathbf{w}$ in \eqref{eq:defi:average_mws_rate}, which can be calculated based on the full knowledge of statistical CSI, i.e., PDFs of the SNRs \cite{liu2011asymptotic,holtzman2001asymptotic,choi2015throughput}, as follows.

\subsection{Rate Estimation Using Full Knowledge of Statistical CSI}\label{subsec:pdf-based_rate_estimation}
We note that the values of the MWS weights control the probability that users are selected for every time slot.
Specifically, user $k$'s average rate achieved by the MWS, $r^{\sim\mu(\cdot|\mathbf{w})}_k$, $k=1,\dots,K$, in \eqref{eq:defi:average_mws_rate} can be calculated as
\begin{equation}\label{eq:average_rate_calculated_based_on_pdf}
\begin{aligned}
& r^{\sim\mu(\cdot|\mathbf{w})}_k \\
=& \expt[\Delta_\text{0} Bx_k(t)\log_2\big(1+\phi_k(t)\big)| \mathbf{x}(t) =  \mu\big(\mathbf{s}(t)\big|\mathbf{w}\big),\forall t] \\
 \stackrel{\text{(a)}}{=} & \Delta_\text{0} B \expt[x_k(t)]\expt[\log_2\big(1+\phi_k(t)\big)|x_k(t)=1] \\
= & \Delta_\text{0} B \Pr[x_k(t)=1]\int_\phi q_k(\phi|x_k(t)=1)\log_2(1+\phi)\textrm{d} \phi \\
\stackrel{\text{(b)}}{=} & \Delta_\text{0} B \int_\phi q_k(\phi) \Pr[x_k(t)=1|\phi_k(t)=\phi] \log_2(1+\phi)\textrm{d} \phi, \ \forall k , 
\end{aligned}
\end{equation}
where (a) is because $x_k(t)$ is binary and (b) uses Bayes' theorem \cite{liu2011asymptotic}.
Here, $q_k(\phi)$ and $q_k(\phi|x_k(t)=1)$ in \eqref{eq:average_rate_calculated_based_on_pdf} are the PDF of user $k$'s SNRs in all $T$ time slots and in those time slots where this user is scheduled, respectively.
$\Pr[x_k(t)=1|\phi_k(t)=\phi]$ in \eqref{eq:average_rate_calculated_based_on_pdf} is the probability that user $k$ is scheduled for a given SNR value $\phi$ in the $t$-th slot, which is calculated for the MWS $\mu(\cdot|\mathbf{w})$ as 
\begin{equation}\label{eq:probability_x=1|phi}
\begin{aligned}
&\Pr[x_k(t)=1|\phi_k(t)=\phi] \\
\stackrel{\text{(a)}}{=}&  \prod_{j\neq k} \Pr[w_j\log_2\big(1+\phi_j(t)\big)<w_k\log_2\big(1+\phi\big)] \\
= &\prod_{j\neq k} \Pr[\phi_j(t)<(1+\phi)^\frac{w_k}{w_j}-1]\\
= &\prod_{j\neq k} \int_0^{(1+\phi)^\frac{w_k}{w_j}-1} q_j(\psi)\textrm{d} \psi \ , \ \forall k \ , t \ , \mathbf{w} \ ,
\end{aligned}
\end{equation}
where (a) uses the definition of the MWS in \eqref{eq:max_weight_scheduler} that a user occupies a slot when it has the highest weighted spectrum efficiency than other users in this slot. By substituting \eqref{eq:probability_x=1|phi} into \eqref{eq:average_rate_calculated_based_on_pdf}, the average rates of MWSs can be rewritten as
\begin{equation}\label{eq:average_rate_calculated_based_on_pdf_final}
\begin{aligned}
r^{\sim\mu(\cdot|\mathbf{w})}_k=&\Delta_\text{0} B \int_\phi q_k(\phi)\Big(\prod_{j\neq k}\int_0^{(1+\phi)^\frac{w_k}{w_j}-1} q_j(\psi)\textrm{d} \psi \Big) \\
&\qquad\qquad\qquad \cdot \log_2(1+\phi)\textrm{d} \phi,\ \forall k , \ \mathbf{w} \ . 
\end{aligned}
\end{equation}
Here, $r^{\sim\mu(\cdot|\mathbf{w})}_k$ can be calculated if PDFs of all users' SNRs are known. 
Such an approach for the calculation of $\mathbf{r}^{\sim\mu(\cdot|\mathbf{w})}$ suffers from two issues. First, estimating PDFs of the SNRs has a high time complexity, or in other words, it requires a large number of time slots to collect samples of channel states such that each state is counted sufficient times. Second, even though the PDFs are obtained, it is difficult to numerically calculate the integrals in \eqref{eq:average_rate_calculated_based_on_pdf_final}, and additional estimators of the integrals are required \cite{liu2011asymptotic}. These issues in PDF-based rate estimation can hardly be addressed. Thus, we propose a new method to estimate the average rates in the sequel.

\section{Rate Estimation \\ Using Mean and Variance of SNRs}
This section proposes a new method for calculating the average rates $\mathbf{r}^{\sim\mu(\cdot|\mathbf{w})}$ in \eqref{eq:defi:average_mws_rate} achieved by the MWS $\mu(\cdot|\mathbf{w})$ in order to solve \eqref{eq:prob:stochastic_weight_optimization_mu_w}.
To estimate the average rates achieved by the MWS, we need the information on how good is each user's channel, which can be expressed by the mean of the SNRs. Further, since the MWS is an opportunistic scheduler, it schedules the user when the user's channel is relatively better than other users (e.g., the user has the highest weighted instantaneous bit rate), as shown in \eqref{eq:max_weight_scheduler}. Note that channel states (or SNRs) are random numbers. Therefore, to accurately estimate the rates, we need to characterize the randomness of the channel states, e.g., using the variance of SNRs.
Here, we denote the mean and the variance of user $k$' SNR $\phi_k(t)$ as
\begin{equation}\label{eq:defi:m_v_snr}
\begin{aligned}
m^{\phi}_{k}   \triangleq \expt[\phi_k(t)]\ , \
v^{\phi}_{k}   \triangleq \vari[\phi_k(t)] \ , \forall k\ .
\end{aligned}
\end{equation}

The feasible rate region $\mathcal{F}$ contains all possible average rates $\mathbf{r}=[r_1,\dots,r_K]^{\rm T}$ achieved by any valid schedulers as
\begin{equation}\label{eq:defi:feasible_rate_region}
\begin{aligned}
&\mathcal{F}\triangleq\Big\{\mathbf{r}| \forall k \ , r_k = \lim_{T\to \infty} \frac{1}{T} \sum_{t=1}^{T} x_k(t)\Delta_\text{0} B \log_2\big(1+\phi_k(t)\big), \\ & \qquad\qquad\qquad\qquad\qquad\qquad \text{s.t.}\  \eqref{eq:const:channel_access_binary}, \ \eqref{eq:const:channel_access} \Big\}\ ,
\end{aligned}
\end{equation}
which is a compact convex set based on justification in \cite{agrawal2002optimality}. 
Note that $\mathbf{r}^{\sim\mu(\cdot|\mathbf{w})}$ are feasible rates that belong to $\mathcal{F}$.
As we cannot use \eqref{eq:average_rate_calculated_based_on_pdf_final} directly, we formulate a rate estimation problem that calculates $\mathbf{r}^{\sim\mu(\cdot|\mathbf{w})}$ based on the feasible rate region $\mathcal{F}$ as
\begin{cor}\label{corollary:mws_as_solution_of_weighting_problem}
For any MWS $\mu(\cdot|\mathbf{w})$ defined in \eqref{eq:max_weight_scheduler}, its achieved average rates in \eqref{eq:defi:average_mws_rate} satisfy
\begin{equation}\label{eq:corollary:mws_as_solution_of_weighting_problem}
\begin{aligned}
\mathbf{r}^{\sim\mu(\cdot|\mathbf{w})} = \arg\max_{\mathbf{r}}\langle\mathbf{w},\mathbf{r}\rangle\ , \ \text{s.t.} \ \mathbf{r} \in \mathcal{F} \ .
\end{aligned}
\end{equation}
\end{cor}
\begin{proof}
The proof is in the appendix.
\end{proof}

We aim to use the mean and variance of the SNRs, defined in \eqref{eq:defi:m_v_snr}, to construct a convex set $\mathcal{G}$ that represents $\mathcal{F}$. In this way, we can find an estimation of $\mathbf{r}^{\sim\mu(\cdot|\mathbf{w})}$ via solving the optimization problem at the right-hand side (RHS) of \eqref{eq:corollary:mws_as_solution_of_weighting_problem}, in which $\mathcal{G}$ replaces $\mathcal{F}$. 
To achieve this, we will construct convex expressions\footnote{We will focus on explaining our methods without the explicit proof of the claim on the convexity of inequalities, equalities and optimization problems. If interested, readers can use the checker implemented in \cite{diamond2016cvxpy} to validate the claimed convexity.} as constraints based on the mean and variance of the SNRs to govern randomness of $\mathbf{x}(t)$, $\mathbf{s}(t)$ and the correlation between them, which also specify the range of the value of users' feasible average rates in $\mathcal{F}$.

\subsection{Bounding the Feasible Rate Region $\mathcal{F}$}\label{sec:bounding_set}
First, let us consider any feasible values of the average rate of user $k$, $r_k$, $k=1,\dots,K$, including the average rates achieved by any MWS in \eqref{eq:defi:average_mws_rate}.
We note that the average rates of users are non-negative, which is expressed in a constraint as
\begin{equation}\label{eq:const:rate_lower_bound}
\begin{aligned}
r_k\geq0, \ \forall k\ .
\end{aligned}
\tag{C1}
\end{equation}
An upper bound on any feasible values of user $k$'s rate, $r_k$, $k=1,\dots,K$, is obtained as
\footnote{When the channel bandwidth is wide, e.g., typically in high-frequency bands such as sub-6GHz band, millimeter band and THz band, frequency selective fading can happen due to multipath propagation of the radio signals \cite{serghiou2022terahertz,alrabeiah2020deep}. Specifically, assuming the system uses orthogonal frequency-division multiplexing, we need to consider the fading gain on each subcarrier when computing the instantaneous bit rate in the $t$-th slot for user $k$, whose expression is $x_k(t)\Delta_0 B \expt [\log_2 (1+|h_k|^2\phi_k(t))]$, $\forall k,t$ \cite{tulino2010capacity}. Here, $h_k$ is the random normalized fading gain in different subcarriers. Note that the above expectation is computed with respect to the random variable $h_k$. Regardless of the fading model, the amount of instantaneous bits scheduled for the $k$-th user can be upper-bounded as $x_k(t)\Delta_0 B \expt[\log_2(1+|h_k|^2\phi_k(t))]\leq x_k(t)\Delta_0 B \log_2(1+\expt[|h_k|^2]\phi_k(t)) = x_k(t)\Delta_0 B \log_2(1+\phi_k(t)) $, using Jensen's inequality. This derivation will lead to the same upper bound of users' data rates. Thus, the proposed methods still apply, while how tight the rate estimations for different channel models are requires further investigation.}
\begin{equation}\label{eq:rate_concave}
\begin{aligned}
r_k  = & \Delta_\text{0} B \expt[x_k(t)\log_2\big(1+\phi_k(t)\big)] \\
=& \Delta_\text{0} B \expt[x_k(t)]\expt[\log_2\big(1+\phi_k(t)\big)|x_k(t)=1] \\
\stackrel{\text{(a)}}{\leq} & \Delta_\text{0}B\expt[x_k(t)]\log_2\big(1+\expt[\phi_k(t)|x_k(t) = 1]\big)\ , \ \forall k \ ,
\end{aligned}
\end{equation}
where (a) uses Jensen's inequality on the expected value of the concave function, $\log_2\big(1+(\cdot)\big)$.
Here, $\expt[\phi_k(t)|x_k(t) = 1]$ in \eqref{eq:rate_concave} is each user's average SNR in their scheduled slots, which can be calculated by
\begin{equation}\label{eq:expt_xpsi}
\begin{aligned}
\expt[\phi_k(t)|x_k(t) = 1]  &=\frac{\expt[x_k(t)\phi_k(t)]}{\expt[x_k(t)]} = \frac{y_k}{p_k} \ , \ \forall k\ ,
\end{aligned}
\end{equation}
where $y_k$ and $p_k$ are defined as
\begin{equation}\label{eq:defi:y_k_and_p_k}
\begin{aligned}
y_k\triangleq\expt[x_k(t)\phi_k(t)]\ ,\ p_k\triangleq\expt[x_k(t)]\ , \ \forall k \ .
\end{aligned}
\end{equation}
Then, we substitute \eqref{eq:expt_xpsi} and \eqref{eq:defi:y_k_and_p_k} into \eqref{eq:rate_concave} as
\begin{equation}\label{eq:const:rate_upper_bound}
\begin{aligned}
r_k \leq\Delta_\text{0}B\cdot p_k\cdot\log_2\big(1+\frac{y_k}{p_k}\big), \ \forall k\ .
\end{aligned}
\tag{C2}
\end{equation}
Note that $p_k$ is the expected value of a binary number $x_k(t)$ defined in \eqref{eq:const:channel_access_binary}, which implies
\begin{equation}\label{eq:const:p_K}
\begin{aligned}
0 \leq p_k \leq 1 \ ,\ \forall k \ .
\end{aligned}
\tag{C3}
\end{equation}
The summation of $x_k(t)$ for $k\in\{1,\dots,K\}$ is less than $1$ at each slot, as stated in \eqref{eq:const:channel_access}, leading to a constraint on $p_k$ as
\begin{equation}\label{eq:const:p_K_sum}
\begin{aligned}
\sum_{k=1}^{K} p_k \leq 1 \ .
\end{aligned}
\tag{C4}
\end{equation}

Next, we study the relationship between $y_k$ and $p_k$, $k=1,\dots,K$, in the covariance matrix of binary indicators of user selection decisions $x_k(t)$, and SNRs of users $\phi_k(t)$, $k=1,\dots,K$, as
\begin{equation}\label{eq:covariance_matrix}
\begin{aligned}
\mathbf{H} \triangleq \begin{bmatrix}
\mathbf{H}^{xx}                             & \mathbf{H}^{x\phi}    \\
(\mathbf{H}^{x\phi})^{\rm T}                & \mathbf{H}^{\phi\phi}
\end{bmatrix}\ .
\end{aligned}
\end{equation}
$\mathbf{H}$ is a $2K\times2K$ matrix, and each submatrix in $\mathbf{H}$ has $K\times K$ dimension. $(\cdot)^{\rm T}$ here is the transpose of a matrix.
Elements in the diagonal of $\mathbf{H}^{x\phi}$ are the covariance between $x_k(t)$ and $\phi_k(t)$ and are calculated based on $y_k$ and $p_k$ as
\begin{equation}\label{eq:const:covariance_matrix_xpsi}
\begin{aligned}
H^{x\phi}_{k,k} =  y_k - p_k m^{\phi}_{k} \  , \forall k\ ,
\end{aligned}
\tag{C5}
\end{equation}
where $m^{\phi}_{k}$ is the mean of $\phi_k(t)$, as defined in \eqref{eq:defi:m_v_snr}.

The lower-right part of $\mathbf{H}$, $\mathbf{H}^{\phi\phi}$, is the covariance matrix of $\phi_i(t)$ and $\phi_j(t)$, $i,j\in\{1,\dots,K\}$, whose elements are
\begin{equation}\label{eq:const:covariance_matrix_psipsi}
\begin{aligned}
H^{\phi\phi}_{i,j} 
= \begin{cases}
    v^{\phi}_{k} \ , & \text{if }i=j,\ \\
    0 \ , & \text{if }i\neq j,\
\end{cases}\ \forall i\ , j \in \{1,\dots,K\} \ ,
\end{aligned}
\tag{C6}
\end{equation}
where $v^{\phi}_{k}$ is the variance of of $\phi_k(t)$, as defined in \eqref{eq:defi:m_v_snr}.
Note that the SNRs of users are assumed to be independent. Thus, all off-diagonal elements in $\mathbf{H}^{\phi\phi}$ are $0$ in \eqref{eq:const:covariance_matrix_psipsi}.

$\mathbf{H}^{xx}$ is the covariance matrix of $x_i(t)$ and $x_j(t)$, $i,j\in\{1,\dots,K\}$. 
Note that $x_i(t)$ and $x_j(t)$ are binaries. The covariance of two binary random numbers is calculated as 
\begin{equation}\label{eq:covariance_xx}
\begin{aligned}
H^{xx}_{i,j}  &= \cov(x_i(t),x_j(t))= \expt[x_i(t)x_j(t)]-\expt[x_i(t)]\expt[x_j(t)]\\
&= \Pr[x_i(t)=1,x_j(t)=1]-\expt[x_i(t)]\expt[x_j(t)] \ , \\ & \ \forall i\ , j \in \{1,\dots,K\} \ .
\end{aligned}
\end{equation}
As stated in \eqref{eq:const:channel_access}, if $i\neq j$, $x_i(t)$ and $x_j(t)$ cannot be simultaneously equal to $1$ in a given slot $t$, implying that $\Pr[x_i(t)=1,x_j(t)=1]=0$ if $i\neq j$. Also, when $i=j$, we have $\Pr[x_i(t)=1,x_j(t)=1]=\Pr[x_i(t)=1]=p_i$. Therefore, the elements in $\mathbf{H}^{xx}$ in \eqref{eq:covariance_xx} can be written as 
\begin{equation}\label{eq:covariance_matrix_xx}
\begin{aligned}
H^{xx}_{i,j} =   \begin{cases}
    p_i-p^2_i, & \text{if }\ i=j,\\
    -p_ip_j, & \text{if }\ i\neq j,
  \end{cases}\ \forall i\ , j \in \{1,\dots,K\} \ .
\end{aligned}
\end{equation}
Note that the above equalities are not convex constraints. In order to construct convex constraints for $\mathcal{G}$, we can relax the above constraints in \eqref{eq:covariance_matrix_xx} as convex ones as
\begin{equation}\label{eq:const:covariance_matrix_xx}
\begin{aligned}
H^{xx}_{k,k} \leq p_k-p^2_k \ ,\ \forall k \ ,
\end{aligned}
\tag{C7}
\end{equation}
where constraints on the elements in the main diagonal and the off-diagonal of $\mathbf{H}^{xx}$ are changed to inequalities and removed, respectively.
Also, we observe that the summation of all elements in $\mathbf{H}^{xx}$ in \eqref{eq:covariance_matrix_xx} is 
\begin{equation}
\begin{aligned}
\sum_{i=1}^{K}\sum_{j=1}^{K}H^{xx}_{i,j} = \sum_{i=1}^{K}p_i - \left(\sum_{i=1}^{K}p_i\right)^2 \ ,
\end{aligned}
\end{equation}
which can be relaxed as a convex constraint as 
\begin{equation}\label{eq:const:covariance_matrix_xx_sum}
\begin{aligned}
\sum_{i=1}^{K}\sum_{j=1}^{K}H^{xx}_{i,j} \leq \sum_{i=1}^{K}p_i - \left(\sum_{i=1}^{K}p_i\right)^2 \ .
\end{aligned}
\tag{C8}
\end{equation}
The positive semidefiniteness of covariance matrices is the constraint on all elements in $\mathbf{H}$ as
\begin{equation}\label{eq:const:covariance_matrix_sdp}
\begin{aligned}
\mathbf{H} \succeq	0 \ .
\end{aligned}
\tag{C9}
\end{equation}
We note that \eqref{eq:const:rate_lower_bound}, \eqref{eq:const:rate_upper_bound}, \dots, \eqref{eq:const:covariance_matrix_sdp} are all convex constraints.

Finally, the collection of \eqref{eq:const:rate_lower_bound}-\eqref{eq:const:covariance_matrix_sdp} defines a convex set $\mathcal{E}$ that contains all possible values of the tuple $(\mathbf{r},\mathbf{p},\mathbf{y},\mathbf{H})$ as
\begin{equation}\label{eq:set_E}
\begin{aligned}
\mathcal{E}\triangleq \{(\mathbf{r},\mathbf{p},\mathbf{y},\mathbf{H})|
\eqref{eq:const:rate_lower_bound} \text{-}
\eqref{eq:const:covariance_matrix_sdp}
\}\ ,
\end{aligned}
\end{equation}
where tuples' elements are defined as $\mathbf{r} = \{r_1,\dots,r_K\}$, $\mathbf{p} = \{p_1,\dots,p_K\}$, $\mathbf{y} = \{y_1,\dots,y_K\}$ and $\mathbf{H}$ in \eqref{eq:covariance_matrix}.
Also, we define the set of all possible values of $\mathbf{r}$ in all tuples of $\mathcal{E}$ as
\begin{equation}\label{eq:set_G}
\begin{aligned}
\mathcal{G}\triangleq \{\mathbf{r}|\mathbf{r}=\proj_1[(\mathbf{r},\mathbf{p},\mathbf{y},\mathbf{H})], \forall (\mathbf{r},\mathbf{p},\mathbf{y},\mathbf{H}) \in\mathcal{E}\}\ ,
\end{aligned}
\end{equation}
which can be interpreted as a projection of tuples in $\mathcal{E}$ at their coordinates of $\mathbf{r}$. Note that \eqref{eq:const:rate_lower_bound}-\eqref{eq:const:covariance_matrix_sdp} are all convex, which implies the set $\mathcal{E}$ in \eqref{eq:set_E} is convex and further implies the projection of $\mathcal{E}$, $\mathcal{G}$ defined in \eqref{eq:set_G}, is also convex \cite{boyd2004convex}.

\subsection{The Proposed Rate Estimation for MWSs}\label{subsec:cvx_rate_estimation}
By using $\mathcal{G}$ as an approximation of $\mathcal{F}$, we can then rewrite the rate estimation problem in \eqref{eq:corollary:mws_as_solution_of_weighting_problem} that  estimates the average rates of a MWS with weights $\mathbf{w}$ as
\begin{equation}\label{eq:prob:estimation_mws_r_g}
\begin{aligned}
\mathbf{r}^{\sim\mu(\cdot|\mathbf{w})}\approx\mathbf{r}^{\sim\mu(\cdot|\mathbf{w})}_\mathcal{G}\triangleq\arg\max_{\mathbf{r}} \langle\mathbf{w},\mathbf{r}\rangle \ , \ \textrm{s.t.} \ \mathbf{r} \in \mathcal{G} \ ,
\end{aligned}
\end{equation}
where $\mathbf{r}^{\sim\mu(\cdot|\mathbf{w})}_\mathcal{G}$ is the estimated value of $\mathbf{r}^{\sim\mu(\cdot|\mathbf{w})}$ in \eqref{eq:defi:average_mws_rate}.
As $\mathcal{G}$ can be represented by \eqref{eq:const:rate_lower_bound}-\eqref{eq:const:covariance_matrix_sdp} according to \eqref{eq:set_E} and \eqref{eq:set_G}, we can further rewrite the problem in the RHS of \eqref{eq:prob:estimation_mws_r_g} as
\begin{equation}\label{eq:prob:estimation_mws_r}
\begin{aligned}
\max_{\mathbf{r},\mathbf{p},\mathbf{y},\mathbf{H}}\ \langle\mathbf{w},\mathbf{r}\rangle \ , \  
\textrm{s.t.} \ \eqref{eq:const:rate_lower_bound} \text{-}
\eqref{eq:const:covariance_matrix_sdp} \ .
\end{aligned}
\end{equation}
Note that the optimal $\mathbf{r}$ from  \eqref{eq:prob:estimation_mws_r} is equal to the optimal $\mathbf{r}$ from the RHS of \eqref{eq:prob:estimation_mws_r_g}, i.e., $\mathbf{r}^{\sim\mu(\cdot|\mathbf{w})}_\mathcal{G}$ in \eqref{eq:prob:estimation_mws_r_g}.
Since \eqref{eq:const:rate_lower_bound}-\eqref{eq:const:covariance_matrix_sdp} are all convex constraints and the objective function in \eqref{eq:prob:estimation_mws_r}, $\langle\mathbf{w},\mathbf{r}\rangle$, is affine, the optimization problem in \eqref{eq:prob:estimation_mws_r} is a convex optimization problem \cite{boyd2004convex}.
Note that $\mathcal{G}$ must be a bounded set so that the estimated rates in \eqref{eq:prob:estimation_mws_r_g} or \eqref{eq:prob:estimation_mws_r} have meaningful values, or otherwise, they will be infinity and meaningless. We prove the boundedness of $\mathcal{G}$ in the appendix. We refer to $\mathcal{G}$ as the bounding set of $\mathcal{F}$.

\section{Proposed MVWO for MWS Design}\label{sec:dwo_formulation}
By replacing \eqref{eq:defi:average_mws_rate} with \eqref{eq:prob:estimation_mws_r} as the constraint in \eqref{eq:prob:stochastic_weight_optimization_mu_w}, we can then rewrite the optimization of the utility function in \eqref{eq:prob:stochastic_weight_optimization_mu_w} as 
\begin{equation}\label{eq:prob:deterministic_weight_optimization_G}
\begin{aligned}
\max_{\mathbf{w},\mathbf{r},\mathbf{p},\mathbf{y},\mathbf{H}} \   f(\mathbf{r})  ,\
\textrm{s.t.}                   &\ \eqref{eq:weight_non_negative_normalized} \ , &\\
                                &\ \mathbf{r},\mathbf{p},\mathbf{y},\mathbf{H} = \arg\max_{\mathbf{r}',\mathbf{p}',\mathbf{y}',\mathbf{H}'}\ \langle\mathbf{w},\mathbf{r}'\rangle, \\ 
                                &\qquad\qquad\qquad\qquad\quad\textrm{s.t.}\ \eqref{eq:const:rate_lower_bound} \text{-}
\eqref{eq:const:covariance_matrix_sdp} .
\end{aligned}
\tag{P2}
\end{equation}
We use the optimal $\mathbf{w}$ of the above problem to approximate the optimal MWS weights.
Note that \eqref{eq:prob:estimation_mws_r} is embedded in \eqref{eq:prob:deterministic_weight_optimization_G}, resulting in that \eqref{eq:prob:deterministic_weight_optimization_G} has a standard form of BLOPs with known iterative solvers \cite{colson2007overview}.
We refer to the embedded \eqref{eq:prob:estimation_mws_r} as the lower-level problem (LLP) of  \eqref{eq:prob:deterministic_weight_optimization_G} and  \eqref{eq:prob:deterministic_weight_optimization_G} as the upper-level problem (ULP) of the embedded \eqref{eq:prob:estimation_mws_r}, respectively. We refer to \eqref{eq:prob:deterministic_weight_optimization_G} as the MVWO since its formulation only uses the mean and variance of users' SNRs.

\subsection{Issue of Existing Iterative Solvers for the BLOP}\label{subsec:issues_of_existing_iterative_methods}
We study how to solve \eqref{eq:prob:deterministic_weight_optimization_G} using iterative solvers.
We first briefly explain the process of iterative solvers for BLOPs~\cite{colson2007overview,sinha2017review}. 
Let $\mathbf{w}^{(i)}$, $i=1,2,\dots$, denote the weights in the $i$-th iteration of an iterative solver.
In each iteration, the LLP is first solved when the weights in its objective $\langle\mathbf{w},\mathbf{r}'\rangle$ is $\mathbf{w}^{(i)}$. 
Next, the process calculates the weights in the next iteration, $\mathbf{w}^{(i+1)}$.
In the existing iterative solvers \cite{colson2007overview}, the calculation of $\mathbf{w}^{(i+1)}$ is done by solving an additional optimization problem formulated based on the Lagrangian of the LLP, the utility function and the solution of the LLP (i.e., the value of $\mathbf{r}$, $\mathbf{p}$, $\mathbf{y}$ and $\mathbf{H}$ when weights in the LLP is $\mathbf{w}^{(i)}$) in the $i$-th iteration. This optimization problem is typically solved by branch and bound algorithms \cite{colson2007overview,bard1990branch} whose complexity is at least polynomial in terms of the problem size (e.g., the number of users in our case) \cite{zhang1996branch}.
Such methods have a high computational complexity. 
Thus, a low-complexity method is required to update the weights.

\section{Proposed Iterative Solver for MVWO}\label{sec:proposed_iterative_solver}
In this section, we design an iterative solver to solve the MVWO problem in \eqref{eq:prob:deterministic_weight_optimization_G}.
We first define the optimal rates $\mathbf{r}^*$ in the bounding set $\mathcal{G}$ that maximize the objective of \eqref{eq:prob:deterministic_weight_optimization_G}, $f(\mathbf{r})$, as
\begin{equation}\label{eq:defi:r_g_*_g}
\begin{aligned}
\mathbf{r}^*\triangleq\arg\max_{\mathbf{r}'} f(\mathbf{r}') \ , \ \textrm{s.t.} \ \mathbf{r}' \in \mathcal{G} \ .
\end{aligned}
\end{equation}
Based on the above definition, the optimal weights that maximize $f(\mathbf{r})$ in \eqref{eq:prob:deterministic_weight_optimization_G} have the following property.
\begin{cor}\label{cor:optimal_condition_w}
        Suppose $\mathbf{r}^{\sim\mu(\cdot|\mathbf{w})}_\mathcal{G}$ are the estimated rates for given weights $\mathbf{w}$ when solving the LLP in  \eqref{eq:prob:deterministic_weight_optimization_G}. if $|\langle\mathbf{w},\mathbf{r}^{\sim\mu(\cdot|\mathbf{w})}_\mathcal{G}-\mathbf{r}^*\rangle|=0$, then $\mathbf{w}$ are the optimal weights that maximize $f(\mathbf{r})$ in \eqref{eq:prob:deterministic_weight_optimization_G}. In other words, $|\langle\mathbf{w},\mathbf{r}^{\sim\mu(\cdot|\mathbf{w})}_\mathcal{G}-\mathbf{r}^*\rangle|=0$ is the sufficient condition for the optimality of $\mathbf{w}$.
\end{cor}
\begin{proof}
Note that $\mathbf{r}^{\sim\mu(\cdot|\mathbf{w})}_\mathcal{G}$ in the above are the optimal rates that maximize the objective $\langle\mathbf{w},\mathbf{r}\rangle$ in the LLP of  \eqref{eq:prob:deterministic_weight_optimization_G}. Thus, $|\langle\mathbf{w},\mathbf{r}-\mathbf{r}^*\rangle|=0$ implies that $\mathbf{r}^*$ also maximize $\langle\mathbf{w},\mathbf{r}\rangle$ in the LLP of \eqref{eq:prob:deterministic_weight_optimization_G}, i.e., $\mathbf{r}^*=\arg\max_{\mathbf{r}'\in\mathcal{G}} \langle\mathbf{w},\mathbf{r}'\rangle$. Also, according to the definition of $\mathbf{r}^*$ in \eqref{eq:defi:r_g_*_g}, no other rates $\mathbf{r}$ can achieve a higher objective value, i.e., $f(\mathbf{r})\leq f(\mathbf{r}^*)$, $\forall \mathbf{r}\in\mathcal{G}$. This implies $\mathbf{w}$ are the optimal weights as we cannot further optimize weights to achieve better rates (with a higher objective value $f(\mathbf{r})$) in $\mathcal{G}$.
\end{proof}
Using the statement, we design the iteration of weights $\mathbf{w}^{(i)}$, $i=1,2,\dots$, to minimize $|\langle\mathbf{w}^{(i)},\mathbf{r}^{(i)}-\mathbf{r}^*\rangle|$. 
We write estimated rates in the $i$-th iteration as $\mathbf{r}^{(i)} \triangleq \mathbf{r}^{\sim\mu(\cdot|\mathbf{w}^{(i)})}_\mathcal{G}$ for simplicity.
In the remaining of this section, Section {\ref{sec:proposed_iterative_solver}-A} computes $\mathbf{r}^*$ in the optimality condition, e.g., $|\langle\mathbf{w}^{(i)},\mathbf{r}^{(i)}-\mathbf{r}^*\rangle|=0$, and initializes the first iteration's weights $\mathbf{w}^{(1)}$. Section \ref{sec:proposed_iterative_solver}-B computes $\mathbf{r}^{(i)}$ in $|\langle\mathbf{w}^{(i)},\mathbf{r}^{(i)}-\mathbf{r}^*\rangle|$ and the next iteration's weights $\mathbf{w}^{(i+1)}$ in each iteration. Section \ref{sec:proposed_iterative_solver}-C checks the optimality condition, i.e., whether $|\langle\mathbf{w}^{(i)},\mathbf{r}^{(i)}-\mathbf{r}^*\rangle|$ is close enough to $0$ in each iteration.

\subsection{Initialization of the Proposed Iterative Process}
The designed solver first finds the optimal rates $\mathbf{r}^*$ in the bounding set $\mathcal{G}$ that maximize the objective $f(\mathbf{r})$ of \eqref{eq:prob:deterministic_weight_optimization_G} by solving the following problem, 
\begin{equation}\label{eq:defi:r_g_*}
\begin{aligned}
\mathbf{r}^*,\mathbf{p}^*,\mathbf{y}^*,\mathbf{H}^* = \arg\max_{\mathbf{r}',\mathbf{p}',\mathbf{y}',\mathbf{H}'}\ f(\mathbf{r}') \ , \ \textrm{s.t.}\ \eqref{eq:const:rate_lower_bound} \text{-}
\eqref{eq:const:covariance_matrix_sdp} \ .
\end{aligned}
\end{equation}
Here, \eqref{eq:const:rate_lower_bound}-\eqref{eq:const:covariance_matrix_sdp} in \eqref{eq:defi:r_g_*} are used to represent $\mathcal{G}$ in \eqref{eq:defi:r_g_*_g}, as defined in \eqref{eq:set_E} and \eqref{eq:set_G}.
Also, \eqref{eq:defi:r_g_*} is a convex optimization problem because its constraints are convex, and its objective function is to maximize a concave function $f(\cdot)$. Note that \eqref{eq:defi:r_g_*} is only solved once before the following iterative process starts. Then,  we initialize the weights in the first iteration, $\mathbf{w}^{(1)}$, as
\begin{equation}\label{eq:weights_first_iteration}
\begin{aligned}
\mathbf{w}^{(1)} \triangleq [\frac{1}{\sqrt{K}},\dots,\frac{1}{\sqrt{K}}]^{\rm T} \ .
\end{aligned}
\end{equation}

\subsection{Low-Complexity Iterative Updates of Weights}\label{subsec:step_calculation}
Next, the LLP for the given weights in the $i$-th iteration, $\mathbf{w}^{(i)}$, is solved as
\begin{equation}\label{eq:lower_lvl_per_iteration}
\begin{aligned}
\mathbf{r}^{(i)},\mathbf{p}^{(i)},\mathbf{y}^{(i)},\mathbf{H}^{(i)} = \arg\max_{\mathbf{r}',\mathbf{p}',\mathbf{y}',\mathbf{H}'}\ \langle\mathbf{w}^{(i)},\mathbf{r}'\rangle\ \textrm{s.t.}\ \eqref{eq:const:rate_lower_bound} \text{-}\eqref{eq:const:covariance_matrix_sdp} .
\end{aligned}
\end{equation}
We construct $\mathbf{u}^{(i)}$ as the weights and normalize $\mathbf{u}^{(i)}$ to get the next iteration's weights $\mathbf{w}^{(i+1)}$, ensuring the normalization requirement in \eqref{eq:weight_non_negative_normalized}. 
$\mathbf{u}^{(i)}$ consists of two parts.
The first part is $a^{(i)}\mathbf{w}^{(i)}+(\mathbf{r}^*-\mathbf{r}^{(i)})$, which is designed to reduce $|\langle\mathbf{w}^{(i+1)},\mathbf{r}^{(i+1)}-\mathbf{r}^*\rangle|$ in the next iteration. Specifically, we want to find the next iteration's weights satisfying $|\langle\mathbf{w}^{(i+1)},\mathbf{r}^{(i+1)}-\mathbf{r}^*\rangle| = 0$, i.e., satisfying the optimality condition. However, $\mathbf{r}^{(i+1)}$ is only computed in the next iteration $i+1$ and is unknown in the current iteration $i$.
Thus, we assume that $\mathbf{r}^{(i)}$ is around $\mathbf{r}^{(i+1)}$ and set
$\mathbf{w}^{(i+1)}$ such that $|\langle\mathbf{w}^{(i+1)},\mathbf{r}^{(i)}-\mathbf{r}^*\rangle|\approx 0$. This is done by setting $a^{(i)}$ to $\|\mathbf{r}^*-\mathbf{r}^{(i)}\|^2_2/\langle\mathbf{w}^{(i)},\mathbf{r}^{(i)}-\mathbf{r}^*\rangle$, as such the inner product between the first part and $(\mathbf{r}^{(i)}-\mathbf{r}^*)$ is $0$, approximating the optimality condition. 
The second part is $b^{(i)}\mathbf{w}^{(i)}$ ensuring that $\mathbf{u}^{(i)}$ is a positive vector (as we only allow the MWS weights to be positive in \eqref{eq:weight_non_negative_normalized}). We achieve this by setting $b^{(i)}$ as the minimum ratio between $(r^*_k-r^{(i)}_k)$ and $w^{(i)}_k$, $k=1,\dots,K$. Consequently, $(\mathbf{r}^*-\mathbf{r}^{(i)}) + b^{(i)}\mathbf{w}^{(i)}$ is a positive vector and so is $a^{(i)}\mathbf{w}^{(i)}+ (\mathbf{r}^*-\mathbf{r}^{(i)}) + b^{(i)}\mathbf{w}^{(i)}$ (as proved later in Corollary \ref{corollary:feasibility_w_per_iteration}). In summary, $\mathbf{u}^{(i)}$ and $\mathbf{w}^{(i+1)}$ are constructed as
\begin{equation}\label{eq:u_per_iteration_definition}
\begin{aligned}
\mathbf{u}^{(i)} \triangleq \underbrace{a^{(i)}\mathbf{w}^{(i)}+(\mathbf{r}^*-\mathbf{r}^{(i)})}_{\text{first part}}+\underbrace{b^{(i)}\mathbf{w}^{(i)}}_{\text{second part}}\ ; \ \mathbf{w}^{(i+1)} =  \frac{\mathbf{u}^{(i)}}{\|\mathbf{u}^{(i)}\|_2} .
\end{aligned}
\end{equation}
In \eqref{eq:u_per_iteration_definition}, $\mathbf{r}^{*}$ and $\mathbf{r}^{(i)}$ are from \eqref{eq:defi:r_g_*} and \eqref{eq:lower_lvl_per_iteration}, respectively, and 
$a^{(i)}$ and $b^{(i)}$ are configured as
\begin{equation}\label{eq:a_per_iteration_definition}
\begin{aligned}
a^{(i)} \triangleq \frac{\|\mathbf{r}^*-\mathbf{r}^{(i)}\|^2_2}{\langle\mathbf{w}^{(i)},\mathbf{r}^{(i)}-\mathbf{r}^*\rangle} , \ b^{(i)}\triangleq -\min\{(\mathbf{r}^*-\mathbf{r}^{(i)})\oslash\mathbf{w}^{(i)}\}  ,
\end{aligned}
\end{equation}
as mentioned before.
The update of weights in \eqref{eq:u_per_iteration_definition} has a linear computational complexity that is lower than the polynomial complexity of solving the additional optimization problem in existing methods \cite{colson2007overview}.
Furthermore, note that the weights in each iteration follow
\begin{cor}\label{corollary:feasibility_w_per_iteration}
If $\mathbf{w}^{(i)}>0$ and $\|\mathbf{w}^{(i)}\|_2=1$, then $\mathbf{w}^{(i+1)}>0$ and $\|\mathbf{w}^{(i+1)}\|_2=1$.
\end{cor}
\begin{proof}
The proof is in the appendix.
\end{proof}
Since the weights in the first iteration, as defined in \eqref{eq:weights_first_iteration}, satisfy the constraints on weights in \eqref{eq:weight_non_negative_normalized}, i.e., $\mathbf{w}^{(1)}>0$ and $\|\mathbf{w}^{(1)}\|_2=1$, the weights in all following iterations satisfy those constraints based on Corollary \ref{corollary:feasibility_w_per_iteration}, i.e., all iterated weights are feasible in the above process.

\subsection{Termination of the Proposed Iterative Process}\label{subsec:termination}

The solver should stop if the iterated weights $\mathbf{w}^{(i)}$ satisfy the optimality condition in Corollary \ref{cor:optimal_condition_w}, i.e., $|\langle\mathbf{w}^{(i)},\mathbf{r}^{(i)}-\mathbf{r}^*\rangle| = 0$.
In practice, if $|\langle\mathbf{w}^{(i)},\mathbf{r}^{(i)}-\mathbf{r}^*\rangle|$ is less than a small positive number $\hat{\epsilon}$, then we stop the iteration and return $\mathbf{w}^{(i)}$ as the optimal weights of \eqref{eq:prob:deterministic_weight_optimization_G}.
We refer to $\hat{\epsilon}$ as the convergence error of the solver, which describes how close $|\langle\mathbf{w}^{(i)},\mathbf{r}^{(i)}-\mathbf{r}^*\rangle|$ is to $0$.
Algorithm \ref{alg:DWO_algorithm} summarizes the process of the iterative solver. The analysis of the convergence of Algorithm \ref{alg:DWO_algorithm} is presented in the following section.
\begin{spacing}{0.8}
\begin{algorithm}[!t]
\caption{Proposed Iterative Solver for MVWO}\label{alg:DWO_algorithm}
\begin{algorithmic}[1]
\STATE Find $\mathbf{r}^*$ via solving \eqref{eq:defi:r_g_*}.
\STATE Initialize $\mathbf{w}^{(1)}$ as \eqref{eq:weights_first_iteration}.
\FOR{$i=1,2,\dots$}
\STATE Construct Problem \eqref{eq:lower_lvl_per_iteration} based on $\mathbf{w}^{(i)}$.
\STATE Find the solution of Problem \eqref{eq:lower_lvl_per_iteration} as  $\mathbf{r}^{(i)}$.
\IF{$|\langle\mathbf{w}^{(i)},\mathbf{r}^{(i)}-\mathbf{r}^*\rangle|<\hat{\epsilon}$}
\STATE \textbf{break and terminate}.
\ELSE
\STATE Calculate $a^{(i)}$ and $b^{(i)}$ as \eqref{eq:a_per_iteration_definition}.
\STATE Calculate $\mathbf{u}^{(i)}$ and $\mathbf{w}^{(i+1)}$ as \eqref{eq:u_per_iteration_definition}.
\ENDIF
\ENDFOR
\STATE \textbf{return} $\mathbf{w}^{(i)}$.
\end{algorithmic}
\end{algorithm}
\end{spacing}

\section{Convergence of Proposed Iterative Solver}\label{sec:convergence_of_iterative_method}
In this section, we prove that Algorithm \ref{alg:DWO_algorithm} converges. Also, we study the computational complexity of the proposed solver. 
\subsection{Convergence of Algorithm \ref{alg:DWO_algorithm}}\label{subsec:convergence_of_alg}
We define the convergence of Algorithm \ref{alg:DWO_algorithm} as that $|\langle\mathbf{w}^{(i)},\mathbf{r}^{(i)}-\mathbf{r}^*\rangle|$ converges to $0$, i.e., for any positive number $\hat{\epsilon}$, $|\langle\mathbf{w}^{(i)},\mathbf{r}^{(i)}-\mathbf{r}^*\rangle|$ is less than $\hat{\epsilon}$ after a number of iterations.
Note that the above convergence defined for Algorithm \ref{alg:DWO_algorithm} also implies that $\mathbf{w}^{(i)}$ converges to the optimal weights that maximize $f(\mathbf{r})$ in \eqref{eq:prob:deterministic_weight_optimization_G}, as discussed in Corollary \ref{cor:optimal_condition_w}.
To prove the convergence of Algorithm \ref{alg:DWO_algorithm}, we study the relationship between the iterated weights $\mathbf{w}^{(i)}$ and the optimal weights of \eqref{eq:prob:deterministic_weight_optimization_G}. 
We collect all optimal weights of \eqref{eq:prob:deterministic_weight_optimization_G} as a set $\mathcal{W}$, i.e., 
\begin{equation}
\begin{aligned}
\mathcal{W} \triangleq \{ \mathbf{w}\big|\ |\langle\mathbf{w},\mathbf{r}^{\sim\mu(\cdot|\mathbf{w})}_\mathcal{G}-\mathbf{r}^*\rangle|=0,\ \text{s.t.}\ \eqref{eq:weight_non_negative_normalized}\} \ ,
\end{aligned}
\end{equation}
and we write $\mathbf{w}^*$ as an arbitrary vector of weights from $\mathcal{W}$, i.e., $\mathbf{w}^*\in \mathcal{W}$ are optimal weights satisfying the optimality condition $|\langle\mathbf{w}^*,\mathbf{r}^{\sim\mu(\cdot|\mathbf{w}^*)}_\mathcal{G}-\mathbf{r}^*\rangle|=0$. The inner product between $\mathbf{w}^{(i)}$ and $\mathbf{w}^*$ has the following properties.
\begin{lemma}\label{lemma:ww_per_iteration}
1) If Algorithm \ref{alg:DWO_algorithm} is not terminating in the $i$-th iteration, then $\forall \mathbf{w}^* \in \mathcal{W}$,
\begin{equation}\label{eq:lemma:ww_monotonic}
\begin{aligned}
    \frac{\langle\mathbf{w}^{(i+1)},\mathbf{w}^*\rangle}{\langle\mathbf{w}^{(i)},\mathbf{w}^*\rangle}
    > \left[1 - \frac{(\langle\mathbf{w}^{(i)},\mathbf{r}^{(i)}-\mathbf{r}^*\rangle)^2}{2\hat{R}^2} \right]^{-\frac{1}{2}}>1   \ .
\end{aligned}
\end{equation}
Here, $\hat{R}$ in \eqref{eq:lemma:ww_monotonic} is a finite number representing the maximum Euclidean distance between any two vectors in $\mathcal{G}$, defined as $\hat{R}\triangleq\max_{\mathbf{r}^\textrm{a},\mathbf{r}^\textrm{b} \in \mathcal{G}} \|\mathbf{r}^\textrm{a} - \mathbf{r}^\textrm{b}\|_2$.
2) $\frac{1}{\sqrt{K}}\leq\langle\mathbf{w}^{(i)},\mathbf{w}^*\rangle\leq1$, $\forall i$ and $\forall \mathbf{w}^*\in \mathcal{W}$.
\end{lemma}
\begin{proof}
The proof can be found in the appendix.
\end{proof}
Lemma \ref{lemma:ww_per_iteration} shows that $\langle\mathbf{w}^{(i)},\mathbf{w}^*\rangle$ is strictly increasing as $i$ increases.
Using the above statements, we can derive the convergence of Algorithm \ref{alg:DWO_algorithm} as

\begin{theorem}\label{theorm:convergence_mvwo}
Algorithm \ref{alg:DWO_algorithm} converges, i.e., $|\langle\mathbf{w}^{(i)},\mathbf{r}^{(i)}-\mathbf{r}^*\rangle|$ converges to $0$ as the number of iterations $i$ increases.
\end{theorem}
\begin{proof}
Define the supremum of $\langle\mathbf{w}^{(i)},\mathbf{w}^*\rangle$ as $o(\mathbf{w}^*)\triangleq\sup_i\langle\mathbf{w}^{(i)},\mathbf{w}^*\rangle$, $\forall \mathbf{w}^*\in \mathcal{W}$,
where $o(\mathbf{w}^*)$ is a finite number less than or equal to $1$ because $\langle\mathbf{w}^{(i)},\mathbf{w}^*\rangle$ is less than or equal to $1$, $\forall i$, as shown in Lemma \ref{lemma:ww_per_iteration}.
We consider following inequalities in the $i$-th iteration that use the definition of $o(\mathbf{w}^*)$, $\forall \mathbf{w}^*\in\mathcal{W}$,
\begin{equation}\label{eq:ww_to_o_ww}
\begin{aligned}
\langle\mathbf{w}^{(i+1)},\mathbf{w}^*\rangle \leq o(\mathbf{w}^*)\Rightarrow\frac{\langle\mathbf{w}^{(i+1)},\mathbf{w}^*\rangle}{\langle\mathbf{w}^{(i)},\mathbf{w}^*\rangle} \leq \frac{o(\mathbf{w}^*)}{\langle\mathbf{w}^{(i)},\mathbf{w}^*\rangle}\  .
\end{aligned}
\end{equation}
By applying \eqref{eq:lemma:ww_monotonic} of Lemma \ref{lemma:ww_per_iteration} to \eqref{eq:ww_to_o_ww}, we obtain
\begin{equation}\label{eq:theorem:1<ca<ww}
\begin{aligned}
1< \underbrace{\left[1 - \frac{(\langle\mathbf{w}^{(i)},\mathbf{r}^{(i)}-\mathbf{r}^*\rangle)^2}{2\hat{R}^2} \right]^{-\frac{1}{2}}}_{\text{(a)}}< \frac{o(\mathbf{w}^*)}{\langle\mathbf{w}^{(i)},\mathbf{w}^*\rangle} , \forall \mathbf{w}^*\in\mathcal{W} \ .
\end{aligned}
\end{equation}
Note that $\langle\mathbf{w}^{(i)},\mathbf{w}^*\rangle$ converges to $o(\mathbf{w}^*)$ based on the monotone convergence theorem of a sequence, which implies $\frac{o(\mathbf{w}^*)}{\langle\mathbf{w}^{(i)},\mathbf{w}^*\rangle}$ converges to $1$. 
Thus, the value of (a) in \eqref{eq:theorem:1<ca<ww} is between $1$ and a real number converging to $1$. This implies that (a) in \eqref{eq:theorem:1<ca<ww} also converges to $1$, where
\begin{equation}
\begin{aligned}
\big[1 - \frac{(\langle\mathbf{w}^{(i)},\mathbf{r}^{(i)}-\mathbf{r}^*\rangle)^2}{2\hat{R}^2} \big]^{-\frac{1}{2}}\to1,
\Rightarrow |\langle\mathbf{w}^{(i)},\mathbf{r}^{(i)}-\mathbf{r}^*\rangle|  \to 0 ,
\end{aligned}
\end{equation}
which proves the convergence of Algorithm \ref{alg:DWO_algorithm}.
\end{proof}

\subsection{Computational Complexity of the Proposed Iterative Solver}\label{subsec:complexity_of_dwo}
Next, we study the computational complexity of Algorithm \ref{alg:DWO_algorithm}, specifically regarding the number of iterations required in Algorithm \ref{alg:DWO_algorithm} before it converges.
\begin{cor}
Algorithm \ref{alg:DWO_algorithm} converges (i.e., $|\langle\mathbf{w}^{(i)},\mathbf{r}^{(i)}-\mathbf{r}^*\rangle|$ converges to $0$) with a convergence error $\hat{\epsilon}$ in $O(\frac{1}{\hat{\epsilon}^2}K\log K)$ iterations.
\end{cor}
\begin{proof}
Assuming that Algorithm \ref{alg:DWO_algorithm} terminates at the $I$-th iteration, i.e., $|\langle\mathbf{w}^{(i)},\mathbf{r}^{(i)}-\mathbf{r}^*\rangle|$ is less than $\hat{\epsilon}$ at the $I$-th iteration, (for simplicity, we assume that $I$ is greater than $1$), we obtain $|\langle\mathbf{w}^{(i)},\mathbf{r}^{(i)}-\mathbf{r}^*\rangle|\geq\hat{\epsilon}$, $\forall i<I$,
which can be applied to \eqref{eq:lemma:ww_monotonic} as
\begin{equation}\label{}
\begin{aligned}
\frac{\langle\mathbf{w}^{(i+1)},\mathbf{w}^*\rangle}{\langle\mathbf{w}^{(i)},\mathbf{w}^*\rangle}>\left[1 - \frac{\hat{\epsilon}^2}{2\hat{R}^2} \right]^{-\frac{1}{2}} \ ,\ \forall i=1,\dots,I-1 \ .
\end{aligned}
\end{equation}
By multiplying both sides of the above inequalities for $i=1,\dots,I-1$, we obtain
\begin{equation}\label{eq:lemma:ww_monotonic_multiply}
\begin{aligned}
\frac{\langle\mathbf{w}^{(I)},\mathbf{w}^*\rangle}{\langle\mathbf{w}^{(1)},\mathbf{w}^*\rangle}>\left[1 - \frac{\hat{\epsilon}^2}{2\hat{R}^2} \right]^{-\frac{1}{2}(I-1)} \ .
\end{aligned}
\end{equation}
Note that $\langle\mathbf{w}^{(1)},\mathbf{w}^*\rangle$ is greater than or equal to $\frac{1}{\sqrt{K}}$ and $\langle\mathbf{w}^{(I)},\mathbf{w}^*\rangle$ is less than or equal to $1$ according to Lemma \ref{lemma:ww_per_iteration}. By applying the above facts to \eqref{eq:lemma:ww_monotonic_multiply}, we obtain
\begin{equation}\label{}
\begin{aligned}
&\left[1 - \frac{\hat{\epsilon}^2}{2\hat{R}^2} \right]^{-\frac{1}{2}(I-1)}<\frac{\langle\mathbf{w}^{(I)},\mathbf{w}^*\rangle}{\langle\mathbf{w}^{(1)},\mathbf{w}^*\rangle} < \sqrt{K} \ , \\
\Rightarrow & I<\frac{\log K}{-\log\left[1-\frac{\hat{\epsilon}^2}{2\hat{R}^2}\right]} +1 \approx \frac{2\hat{R}^2\log K}{\hat{\epsilon}^2}.
\end{aligned}
\end{equation}
Since $\hat{\epsilon}$ is a small positive constant and $\hat{R}$ is roughly proportional to $\sqrt{K}$, the number of iterations of Algorithm~\ref{alg:DWO_algorithm} required is $O(\frac{1}{\hat{\epsilon}^2}K\log K)$ (or $O(K\log K)$, assuming constant $\hat{\epsilon}$).
\end{proof}

We note that this computational complexity describes the number of iterations required for updating the weights in Algorithm \ref{alg:DWO_algorithm}. Meanwhile, the computational complexity of solving the LLP in \eqref{eq:lower_lvl_per_iteration} in each iteration depends on the specific implementation of the convex optimization solver that typically has a polynomial complexity \cite{bubeck2015convex}. Also, the update of weights in each iteration has a linear computational complexity, as mentioned in Section \ref{subsec:step_calculation}.

\section{Online MVWO Architecture for Varying Mean and Variance of SNRs}\label{sec:online_MWS_archiecture}
In this section, we propose an online architecture to apply our proposed MVWO method in networks with non-stationary wireless channels, e.g., users have mobility and distances between users and BSs change. In such cases, the statistics of SNRs, i.e., the mean and variance of users' SNRs change over time.
As shown in Fig.~\ref{fig:op_sch_online_architecture}, the online MVWO architecture includes the scheduler at the BS and an edge server.
\begin{figure}[!t]
\centering
\includegraphics[scale=0.8]{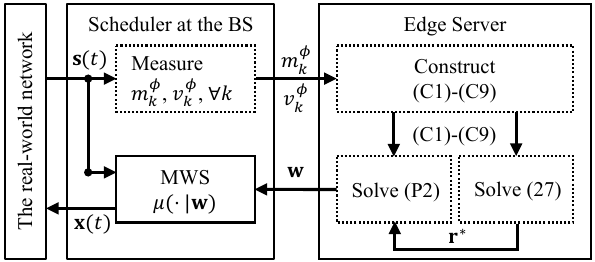}
\caption{The proposed online MVWO architecture.}
\label{fig:op_sch_online_architecture}
\vspace{-0.2cm}
\end{figure}

In each time slot $t$, the scheduler at the BS observes the channel state $\mathbf{s}(t)$ and generates a binary user scheduling action $\mathbf{x}(t)$ according to the MWS defined in \eqref{eq:max_weight_scheduler}.
In real-world networks, e.g., 5G New Radio networks, SNRs can be measured based on the CSI reference signals transmitted with data signals in the wireless channel \cite{3gpp.38.214} and binary user scheduling actions can be mapped into radio resource configurations, including modulation-and-coding schemes and resource block allocations in each transmission, as shown in \cite{gu2021knowledge}.
Meanwhile, the scheduler continuously estimates the mean and variance of users' SNRs by a moving average over the observed values of the SNRs in the past slots in two following steps,
\begin{equation}\label{eq:m_v_measure_online}
\begin{aligned}
m^{\phi}_{k}&\leftarrow \frac{1}{\beta}\sum_{\tau=0}^{\beta-1} \phi_k(t-\tau)\ , \forall k \ ,\\
v^{\phi}_{k}&\leftarrow \frac{1}{\beta}\sum_{\tau=0}^{\beta-1} \big(\phi_k(t-\tau)\big)^2 - (m^{\phi}_{k})^2\ , \forall k \ ,
\end{aligned}
\end{equation}
where $\beta$ is the number of the past time slots used in averaging.
The measured mean and variance of the SNRs are uploaded to the edge server that calculates the weights of the MWS.

The edge server first constructs the constraints in \eqref{eq:const:rate_lower_bound}-\eqref{eq:const:covariance_matrix_sdp} and constantly updates the constraints' parameters based on the latest sent mean and variance of the SNRs from the scheduler. 
Then, the convex optimization problem in \eqref{eq:defi:r_g_*} is constructed and solved to find $\mathbf{r}^*$ defined in \eqref{eq:defi:r_g_*_g} based on \eqref{eq:const:rate_lower_bound}-\eqref{eq:const:covariance_matrix_sdp}. Next, the iterations in lines 3-12 of Algorithm \ref{alg:DWO_algorithm} are executed to solve \eqref{eq:prob:deterministic_weight_optimization_G} based on \eqref{eq:const:rate_lower_bound}-\eqref{eq:const:covariance_matrix_sdp} and $\mathbf{r}^*$, whose return value $\mathbf{w}$ is sent to the scheduler as the weights of the MWS. The above process in the edge server is then repeated until the BS terminates.

The above architecture can also be used in scenarios where the number of users $K$ varies over time.
When a new user joins the network, we need to keep measuring the mean and variance of its SNRs, as shown in \eqref{eq:m_v_measure_online}.
Then, the architecture includes the above mean and variance of the new user when constructing the constraints and computing the rates/weights. Meanwhile, when a user leaves the network, the architecture removes the user's weight, reconstructs the constraints, and recomputes the rates/weights. Note that since this work focuses on the scheduler design using channel statistics, we do not evaluate the architecture for the varying number of users in simulations in the following section.

\section{Evaluation of Proposed Methods}\label{sec:simulation_results}
In this section, we provide the simulation results that evaluate our proposed methods.
\subsection{Simulation Configurations}\label{subsec:simulation_config}
We set $\Delta_\text{0}=1$ (second) and $B=1$ (Hertz) for simplicity in simulations as they linearly scale the average rates while not affecting the performance of our methods. We vary the number of users, $K$, for different cases.
Unless specifically stated, each user's mean of the SNRs, $m_k^{\phi}$, $k=1,\dots,K$, (i.e., the large-scale fading gain) in decibel (dB) is normally distributed with the mean of $10$ dB and the standard deviation of $5$ dB for different episodes \cite{goldsmith2005wireless,tanghe2008industrial} and it remains constant within one episode.
The variance of the SNRs depends on the small-scale fading gains of users, which are i.i.d. in each slot and follow the same normalized Rician distributions with the ratio of the average power in the line-of-sight path to that in the non-line-of-sight paths of $10$ dB \cite{goldsmith2005wireless,tanghe2008industrial}. With the above configurations, SNRs of all users have the variance of $0.17(m_k^{\phi})^2$ or $4.00$ in decimal or in decibel representation, respectively, where $k=1,\dots,K$.

\subsection{Other OS Approaches Compared in Simulation}
\subsubsection{MWS using no prior knowledge of statistical CSI}
We compare our MVWO method with the MWS approaches in \cite{tse1999transmitter,agrawal2002optimality} that can find the optimal MWS's weights to maximize the studied utility function, as defined in \eqref{eq:defi:utility_function_f}. 
Specifically, the weights in these approaches are tuned in every slot as \footnote{The weight update method in \eqref{eq:pf_scheduler} is widely referred to as the proportional fair scheduler \cite{tse2005fundamentals}.}
\begin{equation}\label{eq:pf_scheduler}
\begin{aligned}
\lambda_k \leftarrow (1-\frac{1}{\gamma} )\lambda_k +  \frac{1}{\gamma}x_k(t) \log_2\big(1+\phi_k(t)\big), \
w_k \leftarrow \frac{1}{\lambda_k},\  \forall k ,
\end{aligned}
\end{equation}
where $\lambda_k$ denotes an exponential average of the scheduled instantaneous bit rate of user $k$ (its initial value is set to a small positive number, e.g., $10^{-5}$, to avoid division by zero). 
$\gamma$ in \eqref{eq:pf_scheduler} denotes the size of the exponential average time window, e.g., $100$, $1000$ or $10000$.
Note that the studied utility function is maximized by the above methods when $\gamma$ approaches infinity and the MWS's weights are tuned after sufficient time \cite{tse2005fundamentals}.
Since these approaches use no statistical CSI, we refer to them as statistics-unaware weight optimization (SUWO).
We denote the weights tuned after $\tilde{T}$ slots by using the SUWO methods with $\gamma$ as $\mathbf{w}^{{\rm SUWO}\sim\gamma}_{\tilde{T}}$, where $\tilde{T}$ is varied for different cases and $\mathbf{w}^{{\rm SUWO}\sim\gamma}_{\tilde{T}}$ is normalized after tuning. We denote the average rates achieved by the MWS, $\mu(\cdot|\mathbf{w}^{{\rm SUWO}\sim\gamma}_{\tilde{T}})$, as $\mathbf{r}^{\sim\mu(\cdot|\mathbf{w}^{{\rm SUWO}\sim\gamma}_{\tilde{T}})}$.

\subsubsection{MWS using prior knowledge of the mean of CSI}
Also, a heuristic MWS will be compared, whose weights are designed based on only the mean of the CSI. Note that the studied utility function is a throughput fairness criterion. We can provide fair scheduling decisions for users by setting the weight of user $k$ as the inverse of the average spectrum efficiency, i.e.,
\begin{equation}\label{eq:hf_scheduler}
\begin{aligned}
w_k \leftarrow  \frac{1}{\expt[\log_2(1+\phi_k(t))]}\ ,\  \forall k\ ,
\end{aligned}
\end{equation}
which prevents the MWS from starving those users with low spectrum efficiency.
We refer to the MWS defined by \eqref{eq:hf_scheduler} as a heuristic fairness scheduler (HFS).
We denote the weights in the HFS calculated using the $\tilde{T}$-slot averaged spectrum efficiency as $\mathbf{w}^{{\rm HFS}}_{\tilde{T}}$.

\subsubsection{MDP-based OS using no prior knowledge of statistical CSI}
Additionally, we will compare our method with the MDP-based OS optimized by DRL that uses no statistical CSI.
The reward signal for the studied utility function in every slot is designed in \cite{chen2021bringing} as 
\begin{equation}\label{}
\begin{aligned}
\delta_k(t) &= x_k(t)\log_2(1+\phi_k(t)) \\
&\qquad\qquad\cdot \big[\frac{1}{t}\sum_{\tau=1}^{t} x_k(\tau)\log_2(1+\phi_k(\tau))\big]^{-1} \ , \ \forall k .
\end{aligned}
\end{equation}
The state and the action in the MDP are the channel state and the user scheduling actions in every slot, $\mathbf{s}(t)$ and $\mathbf{x}(t)$, defined in \eqref{eq:defi:s} and \eqref{eq:defi:x}, respectively. 
We use the actor-critic DRL algorithm \cite{chen2021bringing,gu2021knowledge} to train the NN, $\pi(\cdot|\theta)$, as the MDP-based OS, where $\theta$ are the parameters of the NN.
We denote the NN trained after $\tilde{T}$ slots as $\pi(\cdot|\theta^{{\rm DRL}}_{\tilde{T}})$, and its initial values are randomized.

\subsection{Performance of the Proposed Rate Estimation Method}

\begin{figure}[t]
\begin{subfigure}[b]{1\columnwidth}
\centering
\includegraphics[scale=0.85]{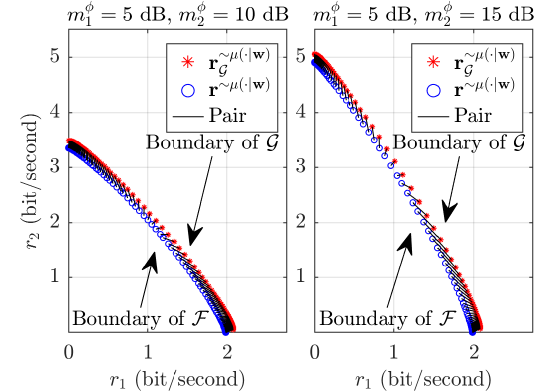}
\caption{Sweeping $\mathbf{w}$ when $K=2$.}
\label{subfig:plot_bs_shape_sigma_50_d}
\end{subfigure}
\begin{subfigure}[b]{1\columnwidth}
\centering
\includegraphics[scale=0.85]{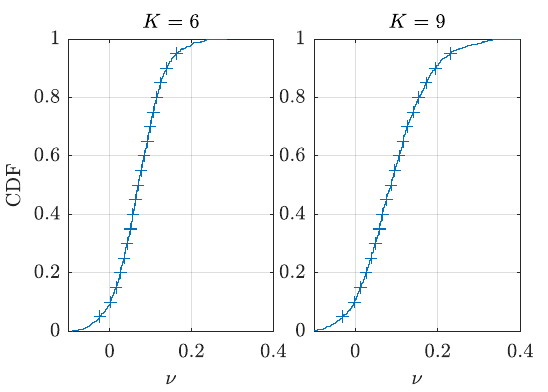}
\caption{$\mathbf{w}=\mathbf{w}^{{\rm SUWO}\sim\gamma}_{\tilde{T}}$ when $K=6$ or $9$.}
\label{fig:plot_n_user_ub_pf_cdf}
\end{subfigure}
\caption{The average rates achieved by $\mu(\cdot|\mathbf{w})$ and their estimated values in \eqref{eq:prob:estimation_mws_r_g} for different $\mathbf{w}$.}\label{fig:plot_approximated_rates}
\vspace{-0.2cm}
\end{figure}

We first compare the estimated value of the average rates scheduled by $\mu(\cdot|\mathbf{w})$ in \eqref{eq:prob:estimation_mws_r_g}, $\mathbf{r}^{\sim\mu(\cdot|\mathbf{w})}_{\mathcal{G}}$, and their measured value, $\mathbf{r}^{\sim\mu(\cdot|\mathbf{w})}$, for two users, i.e., $K=2$. 
Two users' weights are varied as $w_1 = \sin(0.005i\cdot\frac{\pi}{2})$ and $w_2= \cos(0.005i\cdot\frac{\pi}{2})$ for $i=1,\dots,199$, which are all feasible weights, i.e., $\mathbf{w}>0$ and $\|\mathbf{w}\|_2=1$ for all $i$. In each case of $\mathbf{w}$,
$\mathbf{r}^{\sim\mu(\cdot|\mathbf{w})}_{\mathcal{G}}$ is calculated by solving the rate estimation problem in \eqref{eq:prob:estimation_mws_r_g}, and $\mathbf{r}^{\sim\mu(\cdot|\mathbf{w})}$ is measured by averaging the scheduled instantaneous rates in one episode with $10^{5}$ slots.
In Fig. \ref{subfig:plot_bs_shape_sigma_50_d}, the estimated average rates scheduled by the MWS and their actual value measured from the simulation with the same $\mathbf{w}$ (with the legends ``$\mathbf{r}^{\sim\mu(\cdot|\mathbf{w})}_{\mathcal{G}}$'' and ``$\mathbf{r}^{\sim\mu(\cdot|\mathbf{w})}$'', respectively) are connected with a line (with the legend ``Pair'').
Two users have the mean of the SNRs as $m^{\phi}_1=3.16$ and $m^{\phi}_2=10$ in decimal format (or $m^{\phi}_1=5$ and $m^{\phi}_2=10$ in dB), or $m^{\phi}_1=5$ dB and $m^{\phi}_2=15$ dB in Fig. \ref{subfig:plot_bs_shape_sigma_50_d}. The variance of the SNRs follows the configuration in Section \ref{subsec:simulation_config}. 
The results indicate that the estimated and measured average rates at given weights are close to each other. 
Also, note that they form two boundaries of  $\mathcal{F}$ and $\mathcal{G}$, respectively, due to the structure of \eqref{eq:corollary:mws_as_solution_of_weighting_problem} and \eqref{eq:prob:estimation_mws_r_g} (e.g., maximization of the weighted sum of a vector that belongs to a convex set) according to \cite{miettinen2012nonlinear}. 
Since the two boundaries show the same shape and are close to each other, this implies that $\mathcal{G}$ is a close approximation of $\mathcal{F}$. To further illustrate the tightness of the approximation, we manually design some SNR distributions and identify the scenarios where the approximation is not close in the appendix.

Next, we use the optimal rates in the bounding set, $\mathbf{r}^*$ defined in \eqref{eq:defi:r_g_*_g}, to estimate the average rates of MWSs, $\mathbf{r}^{\sim\mu(\cdot|\mathbf{w}^{{\rm SUWO}\sim\gamma}_{\tilde{T}})}$, which is designed by the SUWO methods \cite{tse1999transmitter,agrawal2002optimality}.
We vary the number of users, $K$, as $6$ and $9$ and set $\gamma=10000$ and $\tilde{T}=1000K$, which are sufficiently large for the SUWO methods to find the weights that achieve optimal rates.
For each case of $K$, we run $1000$ episodes where users' SNRs are configured as stated in Section \ref{subsec:simulation_config} and  $\mathbf{r}^{\sim\mu(\cdot|\mathbf{w}^{{\rm SUWO}\sim\gamma}_{\tilde{T}})}$ is measured in $10^5$ slots of each episode.
Note that the $k$-th element of $\mathbf{r}^*$ and $\mathbf{r}^{\sim\mu(\cdot|\mathbf{w}^{{\rm SUWO}\sim\gamma}_{\tilde{T}})}$,  ${r}^*_k$ and ${r}^{\sim\mu(\cdot|\mathbf{w}^{{\rm SUWO}\sim\gamma}_{\tilde{T}})}_k$, are the estimated value and the measured value of user $k$'s average rate achieved by $\mu(\cdot|\mathbf{w}^{{\rm SUWO}\sim\gamma}_{\tilde{T}})$, respectively, $k=1,\dots,K$, and all users are equivalent to each other. 
Thus, we only compare the first user's estimated and measured average rate, e.g., ${r}^*_1$ and ${r}^{\sim\mu(\cdot|\mathbf{w}^{{\rm SUWO}\sim\gamma}_{\tilde{T}})}_1$, in terms of the ratio of the difference in the estimated and measured values to the measured value as $\nu\triangleq 
(r^*_1-r^{\sim\mu(\cdot|\mathbf{w}^{{\rm SUWO}\sim\gamma}_{\tilde{T}})}_1)/r^{\sim\mu(\cdot|\mathbf{w}^{{\rm SUWO}\sim\gamma}_{\tilde{T}})}_1$ whose cumulative distribution function (CDF) is shown in Fig. \ref{fig:plot_n_user_ub_pf_cdf}.
The results indicate that over $90\%$ of the estimated values of the average rates of MWSs are overshoot, or in other words, are bigger than the measured ones, e.g., $\nu>0$.
The results also indicate that the estimated average rates differ from the measured ones by approximately $0\sim20\%$, which implies that the estimated average rates of MWSs in the proposed method are close to their measured values.

\subsection{Evaluation on the Convergence of Algorithm \ref{alg:DWO_algorithm}}
\begin{figure}[t]
\begin{subfigure}[b]{1\columnwidth}
\centering
\includegraphics[scale=0.85]{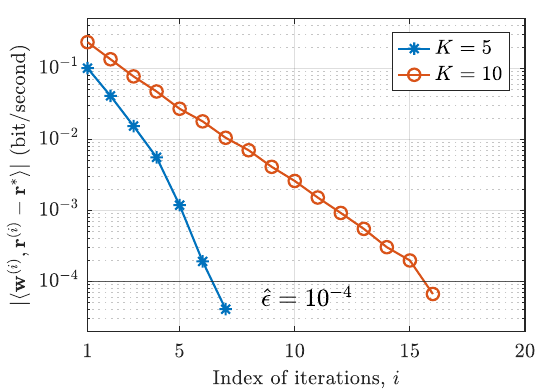}
\caption{The convergence of $|\langle\mathbf{w}^{(i)},\mathbf{r}^{(i)}-\mathbf{r}^*\rangle|$.}
\label{fig:plot_wrr_converge}
\end{subfigure}
\begin{subfigure}[b]{1\columnwidth}
\centering
\includegraphics[scale=0.85]{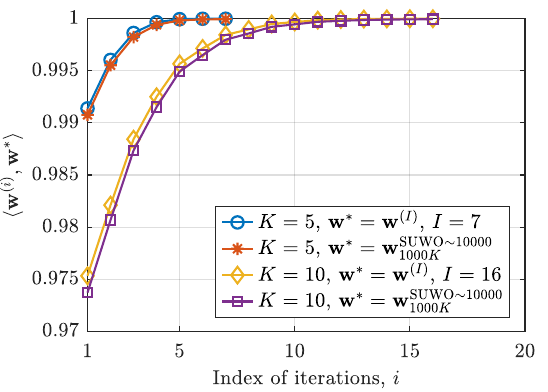}
\caption{The monotone convergence of $\langle\mathbf{w}^{(i)},\mathbf{w}^*\rangle$.}
\label{fig:plot_ww_converge}
\end{subfigure}
\caption{Evaluation of the convergence of Algorithm \ref{alg:DWO_algorithm} when $K=5$ or $10$ and $\hat{\epsilon}=10^{-4}$.}
\vspace{-0.2cm}
\end{figure}

Fig. \ref{fig:plot_wrr_converge} shows the value of $|\langle\mathbf{w}^{(i)},\mathbf{r}^{(i)}-\mathbf{r}^*\rangle|$ in each iteration of Algorithm \ref{alg:DWO_algorithm}, where the number of users are $5$ and $10$ (with the legends ``$K=5$'' and ``$K=10$''. respectively). User $k$'s mean of the SNRs is configured as $m^{\phi}_k=k+5$ dB, $k=1,\dots,K$, and users' variance of the SNRs follows the same configuration as explained in Section \ref{subsec:simulation_config}.
We set $\hat{\epsilon}=10^{-4}$.
With the above configurations, Algorithm \ref{alg:DWO_algorithm} converges in $7$ and $16$ iterations, i.e., $I=7$ and $16$, for $K=5$ and $10$, respectively, as shown in Fig. \ref{fig:plot_wrr_converge}. This validates the proof in Section \ref{subsec:convergence_of_alg}.

Additionally, we validate the monotone convergence of the sequence, $\langle\mathbf{w}^{(i)},\mathbf{w}^*\rangle$, in Fig. \ref{fig:plot_ww_converge}.
Specifically, we show the values of this sequence when we take the weights in the last iteration of the solver, $\mathbf{w}^{(I)}$, as the optimal weights $\mathbf{w}^*$ (with the legend ``$\mathbf{w}^*=\mathbf{w}^{(I)}$''), where $I$ is $7$ and $16$ for $K=5$ and $10$, respectively, as mentioned before.
We also show the values of the sequence when the weights optimized by the SUWO methods \cite{tse1999transmitter,agrawal2002optimality}, $\mathbf{w}^{{\rm SUWO}\sim\gamma}_{\tilde{T}}$, are considered as the optimal weights $\mathbf{w}^*$ (with the legend ``$\mathbf{w}^*=\mathbf{w}^{{\rm SUWO}\sim\gamma}_{\tilde{T}}$''), where $\gamma$ and $\tilde{T}$ are set to sufficiently large values as $10000$ and $1000K$, respectively.
The results indicate that the values of $\langle\mathbf{w}^{(i)},\mathbf{w}^*\rangle$ monotonically increase to $1$ during iterations, which is consistent with the proof in Lemma \ref{lemma:ww_per_iteration}. This also implies that the Euclidean distance between $\mathbf{w}^{(i)}$ and $\mathbf{w}^{{\rm SUWO}\sim\gamma}_{\tilde{T}}$ decreases to $0$ because $\|\mathbf{w}^{(i)}-\mathbf{w}^*\|_2^2=2-2\langle\mathbf{w}^{(i)},\mathbf{w}^*\rangle$.

\begin{figure}[t]
\begin{subfigure}[b]{1\columnwidth}
\centering
\includegraphics[scale=0.85]{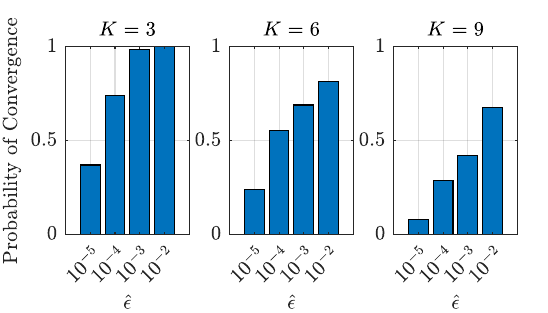}
\caption{Occurrence of convergence in $20$ iterations.}
\label{subfig:plot_eps_convergence_hist_pf_20}
\end{subfigure}
\begin{subfigure}[b]{1\columnwidth}
\centering
\includegraphics[scale=0.85]{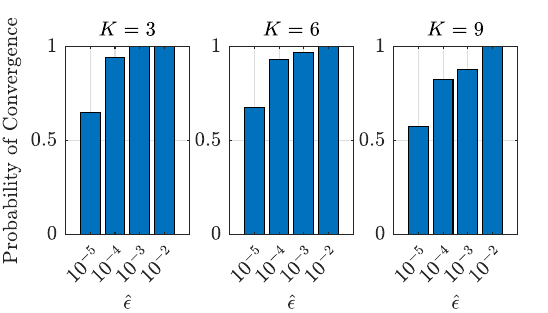}
\caption{Occurrence of convergence in $50$ iterations.}
\label{subfig:plot_eps_convergence_hist_pf_50}
\end{subfigure}
\caption{The probability that the convergence of Algorithm \ref{alg:DWO_algorithm} occurs in given iterations.}
\vspace{-0.2cm}
\end{figure}

Next, we measure the probability that the convergence of Algorithm \ref{alg:DWO_algorithm} occurs within a given number of iterations. We vary the number of users, $K$, as $3$, $6$ and $9$. For each case of $K$, we run Algorithm \ref{alg:DWO_algorithm} for $200$ times where the mean and variance of users' SNRs in each run are randomized, as explained in Section \ref{subsec:simulation_config}.
Figs. \ref{subfig:plot_eps_convergence_hist_pf_20} and \ref{subfig:plot_eps_convergence_hist_pf_50} show the probability that Algorithm \ref{alg:DWO_algorithm} converges in $20$ and $50$ iterations, respectively.
The results indicate that the algorithm is less likely to converge when $K$ is larger, or $\hat{\epsilon}$ is smaller. When the allowed number of iterations increases from $20$ to $50$, the probability of convergence increases significantly. 
The algorithm converges approximately $90\sim100\%$ in $50$ iterations when $\hat{\epsilon}$ is large. 
The above observation complies with the computational complexity analysis in Section \ref{subsec:complexity_of_dwo}.

\subsection{Performance and Time Complexity of the MVWO Method}\label{subsec:simulation_time_complexity_of_dwo}
\begin{figure}[t]
\begin{subfigure}[b]{1\columnwidth}
\centering
\includegraphics[scale=0.85]{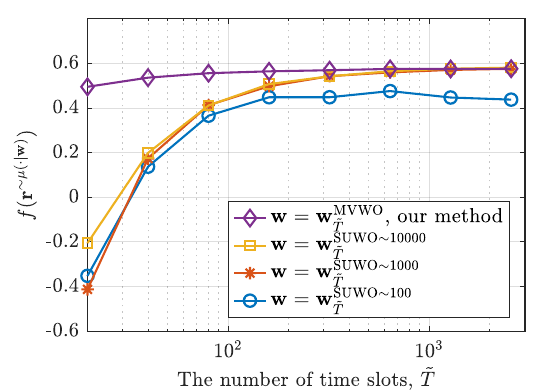}
\caption{$K=3$.}
\label{subfig:plot_3_user_pf_ctime_line}
\end{subfigure}
\begin{subfigure}[b]{1\columnwidth}
\centering
\includegraphics[scale=0.85]{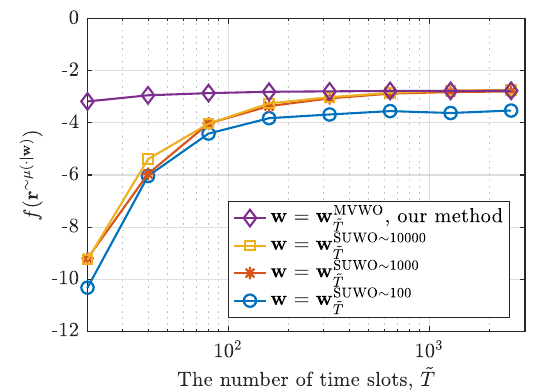}
\caption{$K=6$.}
\label{subfig:plot_6_user_pf_ctime_line}
\end{subfigure}
\caption{Values of the utility function achieved by MWSs optimized by the proposed MVWO method and the SUWO methods when $\tilde{T}$ time slots are used.}
\label{fig:plot_pf_ctime_line}
\vspace{-0.2cm}
\end{figure}
\begin{figure}[!ht]
\begin{subfigure}[b]{1\columnwidth}
\centering
\includegraphics[scale=0.85]{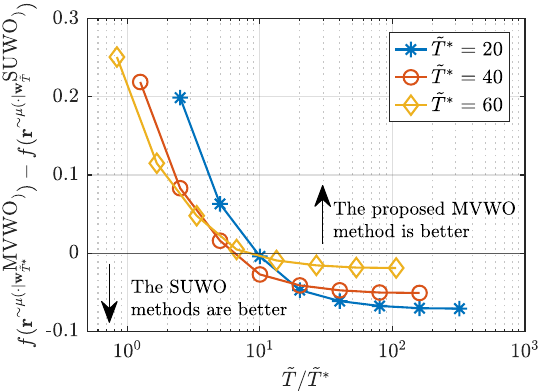}
\caption{$K=3$.}
\label{subfig:plot_3_user_pf_ctime_line_compare}
\end{subfigure}
\begin{subfigure}[b]{1\columnwidth}
\centering
\includegraphics[scale=0.85]{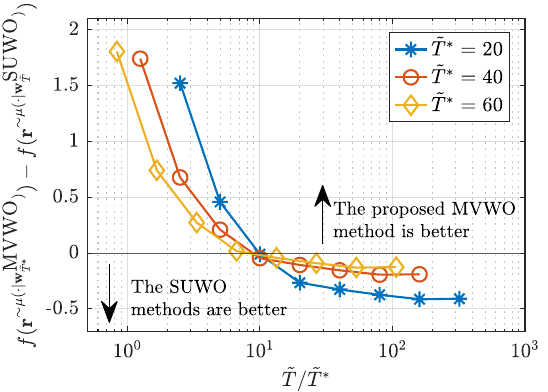}
\caption{$K=6$.}
\label{subfig:plot_6_user_pf_ctime_line_compare}
\end{subfigure}
\caption{Difference in the utility function achieved by MWSs optimized by the proposed MVWO method and the SUWO methods when $\tilde{T}^*$ and $\tilde{T}$ time slots are used, respectively.}
\label{fig:plot_pf_ctime_line_compare}
\vspace{-0.2cm}
\end{figure}

Next, we compare the time complexity (i.e., the number of time slots required) to optimize the weights in the proposed MVWO method and the SUWO methods \cite{tse1999transmitter,agrawal2002optimality}.
We denote the optimal weights found by solving the MVWO in \eqref{eq:prob:deterministic_weight_optimization_G} based on the mean and variance of the SNRs estimated with $\tilde{T}$ slots as $\mathbf{w}^{\rm MVWO}_{\tilde{T}}$. 
Fig. \ref{fig:plot_pf_ctime_line} illustrates the difference in the value of the utility function, $f\big(\mathbf{r}^{\sim\mu(\cdot|\mathbf{w})})$, achieved by different MWSs, $\mu(\cdot|\mathbf{w})$, where the weights, $\mathbf{w}$, are found either by the SUWO methods \cite{tse1999transmitter,agrawal2002optimality} and the proposed MVWO method, (with the legends ``$\mathbf{w}=\mathbf{w}^{{\rm SUWO}\sim\gamma}_{\tilde{T}}$'' and ``$\mathbf{w}=\mathbf{w}^{{\rm MVWO}}_{\tilde{T}}$, our method'', respectively). Here, $\gamma$ are varied as $100$, $1000$ and $10000$ in the SUWO methods. $\hat{\epsilon}$ is set to $10^{-4}$ in the proposed MVWO method. The number of users, $K$, is set as $3$ and $6$ in Figs. \ref{subfig:plot_3_user_pf_ctime_line} and \ref{subfig:plot_6_user_pf_ctime_line}. 
Each point in Fig. \ref{fig:plot_pf_ctime_line} is plotted based on the average value of $f\big(\mathbf{r}^{\sim\mu(\cdot|\mathbf{w})})$ in $100$ episodes. We configure the mean and variance of users' SNRs in each episode as explained in Section \ref{subsec:simulation_config}. 
The values of $\mathbf{r}^{\sim\mu(\cdot|\mathbf{w})}$ are averaged over $10^{5}$ slots for given weights in each episode.
The results in Fig. \ref{fig:plot_pf_ctime_line} indicate that more slots spent in estimating the mean and variance of the SNRs help improve the performance of the weights found by the proposed MVWO method. 
This is because the estimated mean and variance are more accurate when more slots are used.
Results in Fig. \ref{fig:plot_pf_ctime_line} show that the proposed MVWO method converges to the same utility function value as the SUWO methods. Note that the SUWO methods converge to the optimal weights that maximize the utility function \cite{tse2005fundamentals} when $\gamma$ and $\tilde{T}$ are large enough, e.g., $\gamma=1000$ and $\tilde{T}=1000$ in Fig. \ref{fig:plot_pf_ctime_line}. This implies that our method achieves a near-optimal performance with a negligible performance loss.
Also, we observe that the performance of the proposed MVWO method reaches the highest value when it uses approximately $80$ slots, while the performance of the SUWO methods reaches the same value for approximately $320\sim640$ slots.
This indicates that the proposed MVWO method costs $4\sim8$ times fewer system time slots to find the optimal weights. Our method finds optimal weights because the proposed rate approximation method closely estimates the feasible rate region and the average rates for given weights. This estimation accurately represents the system's behaviors and MWSs' performance, and consequently, no online weight adjustment is required.

We keep the same configuration as the above and fix the number of time slots used to estimate the mean and variance of SNRs in our MVWO method as $\tilde{T}^*$, while varying the number of time slots spent in weight tuning in the SUWO methods  \cite{tse1999transmitter,agrawal2002optimality} as $\tilde{T}$, where $\gamma$ is set to $1000$. We measure the difference between the averaged value of the utility function for various ratios of $\tilde{T}$ to $\tilde{T}^*$ when $\tilde{T}^*$ is $20$, $40$ and $60$. 
Figs. \ref{subfig:plot_3_user_pf_ctime_line_compare} and \ref{subfig:plot_6_user_pf_ctime_line_compare} indicate that the proposed MVWO method performs better than the existing SUWO methods when the ratio, $\tilde{T}/\tilde{T}^*$, is less than $10$ and otherwise when the ratio is larger than $10$. This implies that our methods spend $10$ times fewer time slots than the SUWO methods to reach the same performance, and our method's performance is better when the same number of time slots are used in both methods (i.e., when $\tilde{T}/\tilde{T}^*=1$). 

Overall, the simulation results in Fig. \ref{fig:plot_pf_ctime_line} and Fig. \ref{fig:plot_pf_ctime_line_compare} indicate that our method has much lower time complexity than the existing SUWO methods. This is because our method directly uses the measured statistical CSI, while it does not depend on user selection decisions of the MWS. In contrast, the existing SUWO methods require a time average of scheduled instantaneous bit rates, as shown in \eqref{eq:pf_scheduler}, which converges only after each user is scheduled sufficient times.

\begin{figure}[!t]
\centering
\includegraphics[scale=0.85]{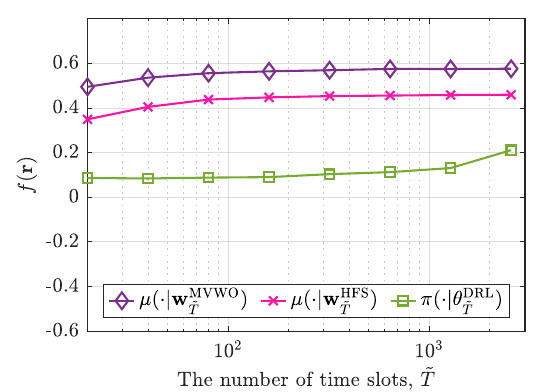}
\caption{Values of the utility function achieved by the MWS optimized by our MVWO method, the HFS in \eqref{eq:hf_scheduler}, and the MDP-based OS optimized by DRL \cite{chen2021bringing} when $K=3$.}
\label{fig:plot_3_user_pf_ctime_line_other_method}
\vspace{-0.2cm}
\end{figure}

In Fig. \ref{fig:plot_3_user_pf_ctime_line_other_method}, we further use the same configuration when $K=3$ and compare the values of the utility function achieved by the MWS optimized by our MVWO method (with the legend ``$\mu(\cdot|\mathbf{w}^{{\rm MVWO}}_{\tilde{T}})$''), the HFS in \eqref{eq:hf_scheduler} (with the legend ``$\mu(\cdot|\mathbf{w}^{{\rm HFS}}_{\tilde{T}})$''), and the MDP-based OS optimized by the DRL method \cite{chen2021bringing} (with the legend ``$\pi(\cdot|\theta^{{\rm DRL}}_{\tilde{T}})$''). The results show that the HFS has a close convergence speed to our method. This is because both methods configure their weights based on the estimation of the statistical CSI. However, the HFS performs worse than our method when both methods converge. This is because our MVWO method exploits the correlation between scheduling actions and channel states based on their first and second moments, which cannot be done by the HFS that only uses the first moment of the channel's statistics. 
The results also show that the DRL method has much slower convergence and performs worse than ours. This is because the NN in the DRL method contains more parameters than MWSs, which can hardly be trained within a short time (e.g., within 1000 time slots). This observation on the time complexity of DRL is consistent with \cite{chen2021bringing}.

\subsection{Performance of Online MVWO Architecture for Varying Mean and Variance of SNRs}
\begin{figure}[!t]
\centering
\includegraphics[scale=0.85]{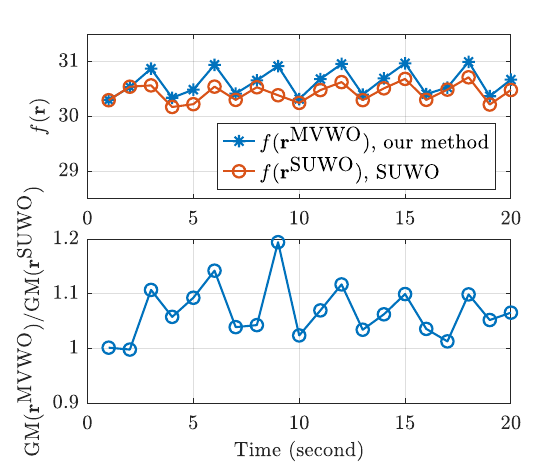}
\caption{The difference in the performance of MWSs optimized by the online MVWO architecture and the SUWO methods \cite{tse1999transmitter,agrawal2002optimality}, where $K=5$.}
\label{fig:plot_online_converge}
\vspace{-0.2cm}
\end{figure}
We then compare the performance of the MVWO method in the proposed online architecture in Section \ref{sec:online_MWS_archiecture} to the SUWO methods \cite{tse1999transmitter,agrawal2002optimality} in the network where the SNRs' mean and variance vary over time. Unlike in the previous case, the scenario in this simulation is closer to real-world networks as users have time-varying large-scale fading due to their mobility.
We assume the number of users, $K$, is $3$, and each user moves $5$ meters per second backward and forward between two points on the ray line from the BS, which are at $20$ and $35$ meters away from the BS. The initial position of user $k$ is at $20+7.5k$ meters away from the BS, $k=1,\dots,K$. The BS's transit power spectrum density is $0$ dBm/Hz, and the noise spectrum density is $-90$ dBm/Hz. The large-scale fading follows a path loss model as $45+30\log_{10}(l)$ dB, where $l$ is the distance between a user and the BS in meters. The small-scale fading is the same as in Section \ref{subsec:simulation_config}.
We use a typical periodicity of user's feedback on SNRs as the duration of a time slot, e.g., $\Delta_\text{0}$ is $10$ milliseconds \cite{3gpp.38.214}, and the bandwidth, $B$, is set as $5$ MHz.
We measure users' average rates every $1$ second (or every 100 slots) during $20$ seconds (or $2000$ time slots). We denote the measured average rates achieved by our method and the SUWO methods \cite{tse1999transmitter,agrawal2002optimality} as $\mathbf{r}^\textrm{MVWO}$ and $\mathbf{r}^\textrm{SUWO}$, respectively. 
Since the users' rates are averaged every 100 slots, the exponential average window, $\gamma$ of \eqref{eq:pf_scheduler}, in the compared SUWO methods is set to the same time scale, e.g., $\gamma=100$ \cite{tse2005fundamentals}. For our method, we set $\beta=20$ in \eqref{eq:m_v_measure_online} and $\hat{\epsilon} = 10^{-4}$ in Algorithm \ref{alg:DWO_algorithm}.
The difference between the utility function in the above two methods can be written as
\begin{equation}\label{eq:f_to_gm}
\begin{aligned}
&f(\mathbf{r}^\textrm{MVWO}) - f(\mathbf{r}^\textrm{SUWO})= \sum_{k=1}^{K}\ln r^\textrm{MVWO}_k -  \sum_{k=1}^{K}\ln r^\textrm{SUWO}_k \\
=&\ln \frac{\prod_{k=1}^{K} r^\textrm{MVWO}_k}{\prod_{k=1}^{K} r^\textrm{SUWO}_k}
=  K  \ln \frac{\textrm{GM}(\mathbf{r}^\textrm{MVWO})}{\textrm{GM}(\mathbf{r}^\textrm{SUWO})} \ ,
\end{aligned}
\end{equation}
where $r^\textrm{MVWO}_k$ and $ r^\textrm{SUWO}_k$ represent user $k$'s rate in $\mathbf{r}^\textrm{MVWO}$ and $\mathbf{r}^\textrm{SUWO}$, respectively.
We quantitatively compare the performance of the two methods in terms of the ratio between the geometric mean of users' average rates, $\textrm{GM}(\mathbf{r}^\textrm{MVWO})/\textrm{GM}(\mathbf{r}^\textrm{SUWO})$, as shown in \eqref{eq:f_to_gm}.
Fig. \ref{fig:plot_online_converge}a shows the value of the utility function every second in our method and the SUWO methods (with the legends ``$f(\mathbf{r}^\textrm{MVWO})$, our method'' and ``$f(\mathbf{r}^\textrm{SUWO})$, SUWO'', respectively), which indicates that our method achieves a higher value of the utility function.
Fig. \ref{fig:plot_online_converge}b illustrates the ratio of the geometric mean of the average rates in two methods, where our method has a $5\sim15\%$ improvement in geometrically averaged rates of users. This is because our MVWO method finds the optimal weights faster than other methods when the channel's statistics change as users move.

\section{Conclusion}
In this paper, we proposed a method to design weights in MWSs by using the limited prior knowledge of statistical CSI. 
Specifically, we computed MWSs' average rates by solving the rate estimation problem based on the mean and variance of users' SNRs.
We formulated the MVWO problem based on the estimated MWSs' rates and proposed an iterative solver, where the iterated weights are proved to converge to the optimal weights.
Also, we designed an online architecture to apply our MVWO method in networks with varying SNRs' mean and variance.
We conducted simulations to validate the accuracy of the rate estimation, the convergence of the proposed solver and the optimality of the weights designed by our MVWO method. Simulations show that our MVWO method consumes $4\sim10$ times fewer time slots in finding the optimal weights and achieves $5\sim15\%$ better average data rates of users than SUWO methods.

\newpage
\section*{Appendix: The Proof of Corollary \ref{corollary:mws_as_solution_of_weighting_problem}}
\begin{proof}
By applying the linearity of the inner product, we obtain
\begin{equation}\label{eq:}
\begin{aligned}
&\langle\mathbf{w},\mathbf{r}^{\sim\mu(\cdot|\mathbf{w})}\rangle \\
=&\lim_{T\to \infty} \frac{1}{T} \sum_{t=1}^{T}\sum_{k=1}^{K} w_k \mu_k\big(\mathbf{s}(t)|\mathbf{w}\big)\Delta_\text{0} B \log_2\big(1+\phi_k(t)\big) \\
\stackrel{\text{(a)}}{\geq}&
\lim_{T\to \infty} \frac{1}{T} \sum_{t=1}^{T} \sum_{k=1}^{K}w_k x_k(t)\Delta_\text{0} B \log_2\big(1+\phi_k(t)\big) \ , \\
&\qquad\qquad\forall \mathbf{x}(t)\ \text{s.t.} \ \eqref{eq:const:channel_access_binary}, \ \eqref{eq:const:channel_access} \ ,
\end{aligned}
\end{equation}
where (a) uses the fact that $\mu(\cdot|\mathbf{w})$ maximizes the weighted sum of instantaneous data rates in each slot, as shown in \eqref{eq:max_weight_scheduler}.
By applying the definition of feasible rate region $\mathcal{F}$ from \eqref{eq:defi:feasible_rate_region} to the above, we have $\langle\mathbf{w},\mathbf{r}^{\sim\mu(\cdot|\mathbf{w})}\rangle\geq \langle\mathbf{w},\mathbf{r}\rangle$, $\forall \mathbf{r}\in \mathcal{F}$.
\end{proof}

\section*{Appendix: Proof of the Boundedness of the Bounding Set}\label{proof:lemma:g_compat_convex}
\begin{proof}
For a given $k\in\{1,\dots,K\}$, let the elements of a vector, $\mathbf{z}$, be $0$ except its $k$-th and $(k+K)$-th element, $z_k$ and $z_{k+K}$.
Then, the positive semidefiniteness of $\mathbf{H}$ leads to
\begin{equation}
\begin{aligned}
&\mathbf{z}^{\rm T}\mathbf{H}\mathbf{z}  \\
= &(z_k)^2 H^{xx}_{k,k}+  2z_kz_{k+K} H^{x\phi}_{k,k} + (z_{k+K})^2 H^{\phi\phi}_{k,k}\geq 0 \ ,\\ 
& \forall (z_k,z_{k+K}) \in \mathbb{R}^2\ , \forall k\ ,
\end{aligned}
\end{equation}
which implies $H^{xx}_{k,k}H^{\phi\phi}_{k,k}\geq (H^{x\phi}_{k,k})^2$, $\forall k$.
By substituting \eqref{eq:const:covariance_matrix_xpsi}\eqref{eq:const:covariance_matrix_psipsi}\eqref{eq:const:covariance_matrix_xx} into the above, we obtain
\begin{equation}\label{eq:upper_bound_of_y}
\begin{aligned}
&p_k(1-p_k)\cdot v^{\phi}_{k} \geq (y_k - p_k m^{\phi}_{k})^2 \ ,\\ \Rightarrow\  &y_k \leq \sqrt{p_k(1-p_k)\cdot v^{\phi}_{k}} +  p_k m^{\phi}_{k} \ ,\forall k \ .
\end{aligned}
\end{equation}
Consider the inequality in \eqref{eq:const:rate_upper_bound}, where we assume  $\Delta_\text{0}=1$ and $B=1$ in order to simplify the notation without generality,
\begin{equation}\label{}
\begin{aligned}
r_k &\leq p_k\log_2(1+\frac{y_k}{p_k})\\
&\stackrel{\text{(a)}}{\leq} p_k\log_2(1+\frac{[\sqrt{p_k(1-p_k)\cdot v^{\phi}_{k}} + p_k m^{\phi}_{k} ]}{p_k})\\
&\stackrel{\text{(b)}}{\leq} p_k\log_2(1+\frac{\sqrt{v^{\phi}_{k}} + m^{\phi}_{k}}{p_k}) \  ,
\end{aligned}
\end{equation}
where (a) uses the inequality in \eqref{eq:upper_bound_of_y} and (b) uses the fact that $p_k\leq1$ and $1-p_k\leq1$ in the numerator of the fraction.
We write $\alpha_k\triangleq \sqrt{v^{\phi}_{k}} + m^{\phi}_{k}$ to simplify the above notation as
\begin{equation}\label{}
\begin{aligned}
r_k &\leq p_k\log_2(1+\frac{\alpha_k}{p_k})\\
&= (p_k+\alpha_k)\log_2(p_k+\alpha_k) \\
&\qquad\qquad\qquad - \alpha_k\log_2(p_k+\alpha_k)- p_k\log_2(p_k), \forall k\ .
\end{aligned}
\end{equation}
Because $(p_k+\alpha_k)\log_2(p_k+\alpha_k)\leq(1+\alpha_k)\log_2(1+\alpha_k)$, $\alpha_k\log_2(p_k+\alpha_k)\geq\alpha_k\log_2(\alpha_k)$ and $p_k\log_2(p_k)\geq\frac{1}{\textrm{e}}\log_2\frac{1}{\textrm{e}}$, we obtain an upper bound on $r_k$ as
\begin{equation}\label{}
\begin{aligned}
r_k \leq (1+\alpha_k)\log_2(1+\alpha_k)- \alpha_k\log_2(\alpha_k)- \frac{1}{\textrm{e}}\log_2\frac{1}{\textrm{e}},  \forall k \ .
\end{aligned}
\end{equation}
Also, $r_k\geq0$, $\forall k$, as shown in \eqref{eq:const:rate_lower_bound}, which implies $r_k$ is bounded $\forall k$. Since all dimensions of $\mathbf{r}$ are bounded, $\mathbf{r}$ is bounded and so is $\mathcal{G}$. 
\end{proof}

\section*{Appendix: Proof of Corollary \ref{corollary:feasibility_w_per_iteration} }\label{proof:corollary:feasibility_w_per_iteration}
We first check the sign of $a^{(i)}$ and $b^{(i)}$. Note that $\mathbf{r}^{(i)}$ is the optimal solution of \eqref{eq:lower_lvl_per_iteration}, which implies $\langle\mathbf{w}^{(i)},\mathbf{r}^{(i)}\rangle>\langle\mathbf{w}^{(i)},\mathbf{r}^*\rangle$. Based on this fact, we can determine that $a^{(i)}>0$.
Suppose $b^{(i)}\leq0$, then $\min\{(\mathbf{r}^*-\mathbf{r}^{(i)})\oslash\mathbf{w}^{(i)}\}\geq0$. This implies $(\mathbf{r}^*-\mathbf{r}^{(i)})\geq0$ and $\langle\mathbf{w}^{(i)},\mathbf{r}^*-\mathbf{r}^{(i)}\rangle\geq0$, which is contradictory to the optimality of $\mathbf{r}^{(i)}$. Therefore, $b^{(i)}>0$.

Based on sign of $a^{(i)}$ and $b^{(i)}$, we can prove that $\mathbf{u}^{(i)}>0$. To achieve this, we first check the sign of the elements in $b^{(i)}\mathbf{w}^{(i)}+(\mathbf{r}^*-\mathbf{r}^{(i)})$, whose $k$-th element is 
\begin{equation}\label{}
\begin{aligned}
&b^{(i)} w^{(i)}_k + r^*_k - r^{(i)}_k =  w^{(i)}_k  (b^{(i)}+ \frac{r^*_k - r^{(i)}_k}{w^{(i)}_k} ) \\
\geq& w^{(i)}  (b^{(i)}+ \min\{(\mathbf{r}^*-\mathbf{r}^{(i)})\oslash\mathbf{w}^{(i)}\} ) = 0 \ , \ \forall k \ .
\end{aligned}
\end{equation}
This implies $b^{(i)}\mathbf{w}^{(i)}+(\mathbf{r}^*-\mathbf{r}^{(i)}) \geq 0 $. By adding $a^{(i)}\mathbf{w}^{(i)}$ to the above, we obtain that $\mathbf{u}^{(i)}>0$. Also, $\mathbf{w}^{(i+1)}$ is normalized $\mathbf{u}^{(i)}$, which implies $\mathbf{w}^{(i+1)}>0$ and $\|\mathbf{w}^{(i+1)}\|_2=1$.

\section*{Appendix: Proof of Lemma \ref{lemma:ww_per_iteration}}\label{proof:lemma:ww_per_iteration}
\begin{proof}

To prove the first statement, we substitute \eqref{eq:u_per_iteration_definition} into $\langle\mathbf{w}^{(i+1)},\mathbf{w}^*\rangle$ as
\begin{equation}\label{eq:proof:ww_inner_product}
\begin{aligned}
&\langle\mathbf{w}^{(i+1)},\mathbf{w}^*\rangle = \langle\frac{a^{(i)}+b^{(i)}}{\|\mathbf{u}^{(i)}\|_2}\mathbf{w}^{(i)}+\frac{1}{\|\mathbf{u}^{(i)}\|_2}(\mathbf{r}^*-\mathbf{r}^{(i)}),\mathbf{w}^*\rangle\\
=&\frac{a^{(i)}+b^{(i)}}{\|\mathbf{u}^{(i)}\|_2}\langle\mathbf{w}^{(i)},\mathbf{w}^*\rangle+\frac{1}{\|\mathbf{u}^{(i)}\|_2}\langle\mathbf{r}^*-\mathbf{r}^{(i)},\mathbf{w}^*\rangle\\
\stackrel{\text{(a)}}{>} &\frac{a^{(i)}+b^{(i)}}{\|\mathbf{u}^{(i)}\|_2}\langle\mathbf{w}^{(i)},\mathbf{w}^*\rangle\ ,
\end{aligned}
\end{equation}
where (a) is because $\langle\mathbf{w}^*,\mathbf{r}^*\rangle>\langle\mathbf{w}^*,\mathbf{r}^{(i)}\rangle$. Also, we have
\begin{equation}\label{eq:alg:v_r_leq_0}
\begin{aligned}
&\langle\mathbf{u}^{(i)},\mathbf{r}^*-\mathbf{r}^{(i)}\rangle=\langle(a^{(i)}+b^{(i)})\mathbf{w}^{(i)}+(\mathbf{r}^*-\mathbf{r}^{(i)}),\mathbf{r}^*-\mathbf{r}^{(i)}\rangle\\
&= \frac{\|\mathbf{r}^*-\mathbf{r}^{(i)}\|^2_2}{\langle\mathbf{w}^{(i)},\mathbf{r}^{(i)}-\mathbf{r}^*\rangle}\langle\mathbf{w}^{(i)},\mathbf{r}^*-\mathbf{r}^{(i)}\rangle\\
&\qquad\qquad\qquad +  \langle b^{(i)}\mathbf{w}^{(i)},\mathbf{r}^*-\mathbf{r}^{(i)}\rangle+\|\mathbf{r}^*-\mathbf{r}^{(i)}\|_2^2\\
&=0+b^{(i)}\langle\mathbf{w}^{(i)},\mathbf{r}^*-\mathbf{r}^{(i)}\rangle\stackrel{\text{(a)}}{<}0 \ ,
\end{aligned}
\end{equation}
where (a) uses the fact that $\langle\mathbf{w}^{(i)},\mathbf{r}^*\rangle<\langle\mathbf{w}^{(i)},\mathbf{r}^{(i)}\rangle$. 
The square of $\ell_2$-norm of $\mathbf{u}^{(i)}$ is
\begin{equation}\label{eq:v_to_a_1}
\begin{aligned}
&\|\mathbf{u}^{(i)}\|_2^2 = \langle\mathbf{u}^{(i)},\mathbf{u}^{(i)}\rangle \\
=&\langle (a^{(i)}+b^{(i)})\mathbf{w}^{(i)}+(\mathbf{r}^*-\mathbf{r}^{(i)}), \mathbf{u}^{(i)}\rangle\\
\stackrel{\text{(a)}}{<}& \langle (a^{(i)}+b^{(i)})\mathbf{w}^{(i)}, \mathbf{u}^{(i)}\rangle\\
=& (a^{(i)}+b^{(i)})^2 + (a^{(i)}+b^{(i)})\langle\mathbf{w}^{(i)},\mathbf{r}^*-\mathbf{r}^{(i)}\rangle \ ,
\end{aligned}
\end{equation}
where (a) uses the inequality in \eqref{eq:alg:v_r_leq_0}.
By dividing $(a^{(i)}+b^{(i)})^2$ at each term in the last inequality, we obtain
\begin{equation}\label{eq:v_to_a}
\begin{aligned}
&\left(\frac{\|\mathbf{u}^{(i)}\|_2}{a^{(i)}+b^{(i)}}\right)^2 < 1 + \frac{\langle\mathbf{w}^{(i)},\mathbf{r}^*-\mathbf{r}^{(i)}\rangle}{a^{(i)}+b^{(i)}}\\
\Rightarrow &\left(\frac{\|\mathbf{u}^{(i)}\|_2}{a^{(i)}+b^{(i)}}\right)^2 < 1 - \frac{(\langle\mathbf{w}^{(i)},\mathbf{r}^{(i)}-\mathbf{r}^*\rangle)^2}{(a^{(i)}+b^{(i)})\langle\mathbf{w}^{(i)},\mathbf{r}^{(i)}-\mathbf{r}^*\rangle},
\end{aligned}
\end{equation}
and substituting $a^{(i)}$ and $b^{(i)}$ of \eqref{eq:a_per_iteration_definition} into the denominator in the RHS of the above
\begin{equation}\label{eq:a_b<r}
\begin{aligned}
&(a^{(i)}+b^{(i)} )\langle\mathbf{w}^{(i)},\mathbf{r}^{(i)}-\mathbf{r}^*\rangle\\
=&  \|\mathbf{r}^*-\mathbf{r}^{(i)}\|_2^2 + \sum_k b^{(i)} w^{(i)}_k(r^{(i)}_k-r^*_k) \\
<  &\|\mathbf{r}^*-\mathbf{r}^{(i)}\|_2^2 + \sum_{k: r^{(i)}_k-r^*_k \leq 0 } \frac{r^{(i)}_k-r^*_k}{w^{(i)}_k}w^{(i)}_k(r^{(i)}_k-r^*_k) \\
&\qquad\qquad + \sum_{k: r^{(i)}_k-r^*_k > 0 } b^{(i)} w^{(i)}_k(r^{(i)}_k-r^*_k) \\
\approx   &\|\mathbf{r}^*-\mathbf{r}^{(i)}\|_2^2 + \sum_k \frac{r^{(i)}_k-r^*_k}{w^{(i)}_k}w^{(i)}_k(r^{(i)}_k-r^*_k) \\
=& 2\|\mathbf{r}^*-\mathbf{r}^{(i)}\|_2^2 \ ,
\end{aligned}
\end{equation}
where we note that both $\mathbf{r}^{(i)}$ and $\mathbf{r}^*$ are vectors in $\mathcal{G}$, which implies $\|\mathbf{r}^*-\mathbf{r}^{(i)}\|_2\leq\hat{R}$ based on the definition of $\hat{R}$.
By applying \eqref{eq:a_b<r} to \eqref{eq:v_to_a}, we can obtain
\begin{equation}\label{eq:v_to_a_2}
\begin{aligned}
&\left(\frac{\|\mathbf{u}^{(i)}\|_2}{a^{(i)}+b^{(i)}}\right)^2 < 1 - \frac{(\langle\mathbf{w}^{(i)},\mathbf{r}^{(i)}-\mathbf{r}^*\rangle)^2}{2\hat{R}^2} < 1 \ ,\\
\Rightarrow& \frac{\|\mathbf{u}^{(i)}\|_2}{a^{(i)}+b^{(i)}} < \left[1 - \frac{(\langle\mathbf{w}^{(i)},\mathbf{r}^{(i)}-\mathbf{r}^*\rangle)^2}{2\hat{R}^2}\right]^{\frac{1}{2}} < 1 \ .
\end{aligned}
\end{equation}
Note that $\hat{R}$ is finite because $\mathcal{G}$ is bounded.
By substituting \eqref{eq:v_to_a_2} in \eqref{eq:proof:ww_inner_product}, we obtain \eqref{eq:lemma:ww_monotonic}. 

For the second statement, the inner product between $\mathbf{w}^{(i)}$ and $\mathbf{w}^*$ is
\begin{equation}
\begin{aligned}
&\langle\mathbf{w}^{(i)},\mathbf{w}^*\rangle
\stackrel{\text{(a)}}{\geq} \langle\mathbf{w}^{(1)},\mathbf{w}^*\rangle
= \frac{1}{\sqrt{K}}\sum_{k=1}^{K}w^*_k\\
&\stackrel{\text{(b)}}{\geq} \frac{1}{\sqrt{K}}\sum_{k=1}^{K}(w^*_k)^2= \frac{1}{\sqrt{K}} \ ,
\end{aligned}
\end{equation}
where (a) uses the fact that $\langle\mathbf{w}^{(i)},\mathbf{w}^*\rangle$ is monotonic increasing based on the first statement in Lemma \ref{lemma:ww_per_iteration} and (b) is because $w_k$ is less than or equal to $1$, $k=1,\dots,K$. Further, by applying Cauchy–Schwarz inequality, we obtain that $\langle\mathbf{w}^{(i)},\mathbf{w}^*\rangle=|\langle\mathbf{w}^{(i)},\mathbf{w}^*\rangle|\leq \|\mathbf{w}^{(i)}\|_2\|\mathbf{w}^*\|_2=1$.

\end{proof}

\section*{Accuracy of Rate Estimation When SNRs'~Variance Increases}
\begin{figure}[!ht]
\centering
\includegraphics[scale=0.8]{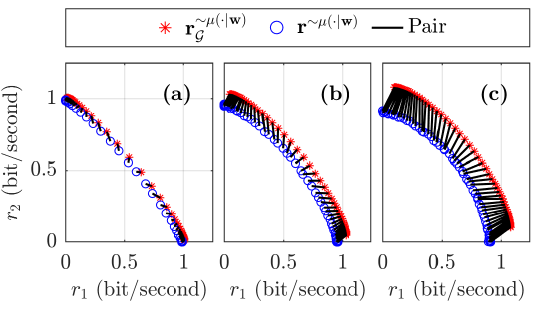}
\caption{The average rates achieved by $\mu(\cdot|\mathbf{w})$ and their estimated values in \eqref{eq:prob:estimation_mws_r_g} for sweeping $\mathbf{w}$ when the SNRs' variance increases. The variance of SNRs is (a) $\exp(0.25^2)-1$, (b) $\exp(0.5^2)-1$ and (c) $\exp(0.75^2)-1$.}
\label{fig:plot_bs_shape_sigma_lgn}
\vspace{-0.2cm}
\end{figure}
We identify that the proposed rate estimation method becomes less tight when the variance of the SNRs increases. To illustrate this effect, we consider a two-user case where both users have the same SNR distribution. To easily adjust the variance, we assume that the SNRs follow a normalized log-normal distribution with a mean of $1$ and a variance of $\exp(\sigma^2)-1$ ($\sigma$ is a parameter determining the variance). We simulate the cases where $\sigma$ is $0.25$, $0.5$ and $0.75$ in Fig. \ref{fig:plot_bs_shape_sigma_lgn}. The results show that when the variance of SNRs increases, the gap between the measured rates and the estimated rates increases. This is because Jenson's equality in \eqref{eq:rate_concave} becomes less tight when the variance of SNRs increases. We note that the simulations here only show the relation between the estimation accuracy and the variance of SNRs, while they do not represent the practical scenarios.

\bibliography{ref}
\bibliographystyle{IEEEtran}

\begin{IEEEbiography}
[{\includegraphics[width=1in,height=1.25in,clip,keepaspectratio]{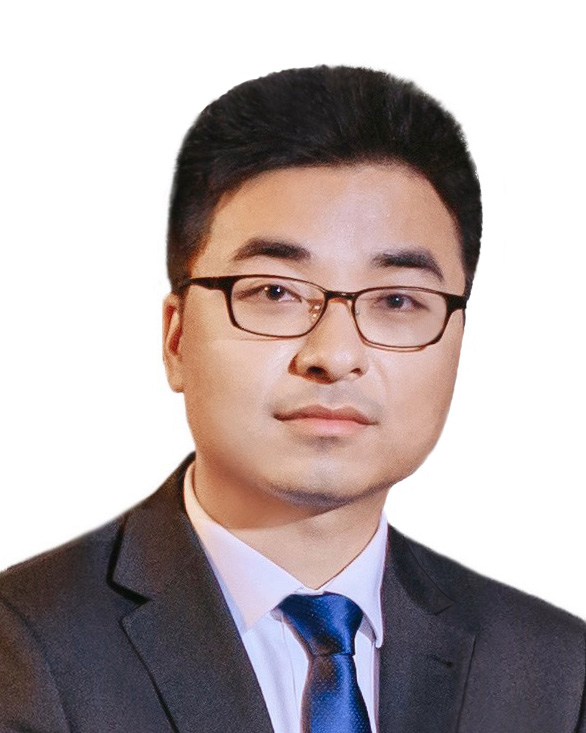}}]{Zhouyou Gu} received the B.E. (Hons.) and M.Phil.
degrees from The University of Sydney (USYD), Australia,
in 2016 and 2019, respectively, where he completed his Ph.D. degree with the School of Electrical and Information Engineering in 2023. 
He was a research assistant at the Centre for IoT and Telecommunications at USYD. His research interests include real-time scheduler designs, programmability, and graph and machine-learning methods in wireless networks.
\end{IEEEbiography}

\begin{IEEEbiography}
[{\includegraphics[width=1in,height=1.25in,clip,keepaspectratio]{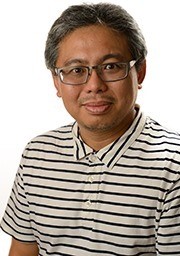}}]{Wibowo Hardjawana} (Senior Member, IEEE) received a PhD in electrical engineering from The University of Sydney, Australia. He is currently a Senior Lecturer in Telecommunications with the School of Electrical and Information Engineering at the University of Sydney. Before that, he was with Singapore Telecom Ltd., managing core and radio access networks. His current fundamental and applied research interests are in AI applications for 5/6G cellular radio access and Wi-Fi networks. He focuses on system architectures, resource scheduling, interference, signal processing, and the development of wireless standard-compliant prototypes. He has also worked with several industries in the area of 5G and long-range Wi-Fi. He was an Australian Research Council Discovery Early Career Research Award Fellow.
\end{IEEEbiography}

\begin{IEEEbiography}
[{\includegraphics[width=1in,height=1.25in,clip,keepaspectratio]{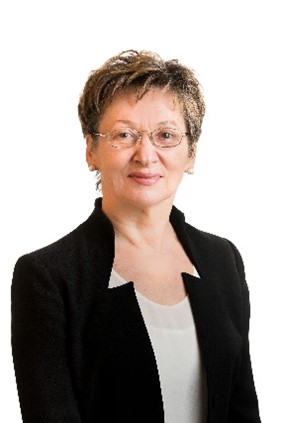}}]{Branka Vucetic} (Life Fellow, IEEE) received the B.S., M.S., and Ph.D. degrees in electrical engineering from the University of Belgrade, Belgrade, Serbia. She is an Australian Laureate Fellow, a Professor of Telecommunications, and the Director of the Centre for IoT and Telecommunications, the University of Sydney, Camperdown, NSW, Australia. Her current research work is in wireless networks and Industry 5.0. In the area of wireless networks, she works on communication system design for 6G and wireless AI. In the area of Industry 5.0, her research is focused on the design of cyber–physical human systems and wireless networks for applications in healthcare, energy grids, and advanced manufacturing. She is a Fellow of the Australian Academy of Technological Sciences and Engineering and the Australian Academy of Science.
\end{IEEEbiography}

\end{document}